\documentclass[a4paper,UKenglish,cleveref, autoref, thm-restate]{lipics-v2021}

\usepackage{xspace, amsthm}
\usepackage{xcolor,stmaryrd}
\usepackage{scalerel,multicol}
\usepackage{mathtools, accents}            
\usepackage{ltl}
\usepackage{tabularx}
\usepackage{multirow}
\usepackage{makecell}
\usepackage{thmtools, comment, url} 
\usepackage{amsmath}
\usepackage{ifthen}

\newboolean{appdx}
\setboolean{appdx}{true}

\newcommand{\LTL}{\protect\ensuremath{\mathrm{LTL}}\xspace}
\newcommand{\CTL}{\protect\ensuremath{\mathrm{CTL}}\xspace}
\newcommand{\QPTL}{\protect\ensuremath{\mathrm{QPTL}}\xspace}
\newcommand{\ML}{\protect\ensuremath{\mathrm{ML}}\xspace}
\newcommand{\teamltl}{\protect\ensuremath{\mathrm{TeamLTL}}\xspace}

\newcommand{\hyctl}{\protect\ensuremath{\mathrm{HyperCTL^*}}\xspace}
\newcommand{\hyltl}{\protect\ensuremath{\mathrm{HyperLTL}}\xspace}
\newcommand{\uhyltl}{\protect\ensuremath{\forall^*\mathrm{HyperLTL}}\xspace}
\newcommand{\HQPTLP}{\protect\ensuremath{\mathrm{HyperQPTL\textsuperscript{\hskip-2pt\small +}}}\xspace}
\newcommand{\HQPTL}{\protect\ensuremath{\mathrm{HyperQPTL}}\xspace}

\newcommand{\dfn}{\mathrel{\mathop:}=}
\newcommand{\ddfn}{\mathrel{\mathop{{\mathop:}{\mathop:}}}=}

\newcommand{\set}[1]{\{ #1 \}}
\newcommand{\ldot}{\mathpunct{.}}
\newcommand{\pow}[1]{2^{#1}}

\newcommand{\tr}{\mathrm{TR}}

\newcommand{\LL}{\mathcal{L}}

\newcommand{\Sub}{\mathrm{Sub}}
\newcommand{\aux}{\varphi_\mathrm{aux}}

\newcommand{\todofy}[1]{\textnormal{\color{blue}\scriptsize+++Fan: #1+++}}
\newcommand{\todojv}[1]{\textnormal{\color{cyan}\scriptsize+++Jonni: #1+++}}
\newcommand{\todojh}[1]{\textnormal{\color{violet}\scriptsize+++Jana: #1+++}}

\newcommand{\U}{\LTLuntil}
\newcommand{\W}{\LTLweakuntil}
\newcommand{\X}{\LTLnext}
\newcommand{\G}{\LTLglobally}

\newcommand{\pathvars}{\mathcal{V}}
\newcommand{\pathassign}{\Pi}

\newcommand{\N}{\mathbb N}
\newcommand{\ap}{\mathrm{AP}}

\newcommand{\traces}{\mathrm{Traces}}
\newcommand{\teams}{\mathrm{TMS}}

\newcommand{\truth}[1]{\llbracket #1 \rrbracket}
\newcommand{\rank}{\mathrm{rank}}
\newcommand{\cons}{\leadsto_\mathrm{c}}
\newcommand{\lcons}{\leadsto_\mathrm{lc}}

\newcommand{\PSPACE}{\protect\ensuremath{\mathrm{PSPACE}}}

\newcommand{\EXPSPACE}{\protect\ensuremath{\mathrm{EXPSPACE}}}

\newcommand{\nlor}{\lor_{\mathrm{NE}}}
\newcommand{\nesub}{\nabla}

\renewcommand{\phi}{\varphi}
\renewcommand{\epsilon}{\varepsilon}
\newcommand{\nats}[0]{\mathbb{N}}

\DeclareMathOperator\dep{\mathrm{dep}}

\newcommand{\kK}{\mathrm{K}}

\newcommand{\single}{\operatorname{singleton}}
\newcommand{\lpreserve}{\operatorname{c_l-preserve}}
\newcommand{\mpreserve}{\operatorname{c_m-preserve}}
\newcommand{\rpreserve}{\operatorname{c_r-preserve}}
\newcommand{\spreserve}{\operatorname{c_s-preserve}}

\newcommand{\ldecrease}{\operatorname{c_l-decrease}}
\newcommand{\mdecrease}{\operatorname{c_m-decrease}}
\newcommand{\rdecrease}{\operatorname{c_r-decrease}}
\newcommand{\sdecrease}{\operatorname{c_s-decrease}}

\newcommand{\clor}{\varovee}
\newcommand{\Llor}{\lor_\mathrm{L}}
\DeclareMathOperator*{\Clor}{\scalerel*{\ovee}{\sum}}

\newcommand{\intimp}{\multimap}
\newcommand{\cneg}{\sim}
\newcommand{\asub}{\mathsf A}

\newcommand{\forallu}{\accentset{u}{\forall}}

\newcommand{\existsu}{\accentset{u}{\exists}}
\newcommand{\quantu}{\accentset{u}{Q}}

\newcommand{\flatop}{\accentset{1}{\asub}}

\newcommand{\forallso}{\forall}
\newcommand{\existsso}{\exists}

\newcommand{\qst}{q^{\mathcal{S}_T}}

\newcommand{\allatoms}{\mathcal{A}_\mathrm{all}}
\newcommand{\dcatoms}{\mathcal{A}_\mathrm{dc}}

\title{Linear-time Temporal Logic with Team Semantics:\\ Expressivity and Complexity} 

\titlerunning{LTL with Team Semantics: Expressivity and Complexity} 

\author{Jonni Virtema}{Institute for Theoretical Computer Science, Leibniz Universit\"at Hannover, Germany \and
Department of Computer Science, University of Sheffield, United Kingdom
}{j.t.virtema@sheffield.ac.uk}{https://orcid.org/0000-0002-1582-3718}{Supported by the DFG grant VI 1045/1-1.}
\author{Jana Hofmann}{CISPA Helmholtz Center for Information Security, Saarbr\"ucken, Germany}{jana.hofmann@cispa.de}{https://orcid.org/0000-0003-1660-2949}{}
\author{Bernd Finkbeiner}{CISPA Helmholtz Center for Information Security, Saarbr\"ucken, Germany}{finkbeiner@cispa.de}{https://orcid.org/0000-0002-4280-8441}{}
\author{Juha Kontinen}{Department of Mathematics and Statistics, University of Helsinki, Finland}{juha.kontinen@helsinki.fi}{https://orcid.org/0000-0003-0115-5154}{Supported by grant 308712 of the Academy of Finland.}
\author{Fan Yang}{Department of Mathematics and Statistics, University of Helsinki, Finland}{fan.yang@helsinki.fi}{https://orcid.org/0000-0003-0392-6522}{Supported by grants 330525 and 308712 of Academy of Finland, and Research Funds of University of Helsinki.}

\authorrunning{J. Virtema, J. Hofmann, B. Finkbeiner, J. Kontinen, F. Yang} 

\Copyright{Jonni Virtema, Jana Hofmann, Bernd Finkbeiner, Juha Kontinen and Fan Yang} 



\category{} 

\relatedversion{}
\ifbool{appdx}{}{\relatedversiondetails{Full Version}{https://arxiv.org/abs/2010.03311}}

\funding{\textit{Jana Hofmann and Bernd Finkbeiner:} Partially supported by the German Research Foundation (DFG) as part of the Collaborative Research Center ``Foundations of Perspicuous Software Systems'' (TRR 248, 389792660) and by the European Research Council (ERC) Grant OSARES (No. 683300).}

\nolinenumbers 

\hideLIPIcs  

\EventEditors{Miko{\l}aj Boja\'{n}czyk and Chandra Chekuri}
\EventNoEds{2}
\EventLongTitle{41st IARCS Annual Conference on Foundations of Software Technology and Theoretical Computer Science (FSTTCS 2021)}
\EventShortTitle{FSTTCS 2021}
\EventAcronym{FSTTCS}
\EventYear{2021}
\EventDate{December 15--17, 2021}
\EventLocation{Virtual Conference}
\EventLogo{}
\SeriesVolume{213}
\ArticleNo{3}

\begin{document}
\allowdisplaybreaks	
	
	\maketitle
	
	\begin{abstract}
	We study the expressivity and complexity
	of model checking
	of linear temporal logic with team semantics (TeamLTL). TeamLTL, despite being a purely modal logic, is capable of defining hyperproperties, i.e., properties which relate multiple execution traces.
	TeamLTL has been introduced quite recently and only few results are known regarding its expressivity and its model checking problem.
	We relate the expressivity of \teamltl to logics for hyperproperties obtained by extending LTL with trace and propositional quantifiers (\hyltl and \HQPTL).
	By doing so, we obtain a number of model checking results for \teamltl and identify its undecidability frontier. In particular, we show decidability of model checking of the so-called \emph{left-flat} fragment of any downward closed \teamltl -extension. Moreover, we establish that the model checking problem of \teamltl with Boolean disjunction and inclusion atoms is undecidable. 
	
\end{abstract}

	\section{Introduction}
Linear-time temporal logic (\LTL) is one of the most prominent logics for the specification and verification of reactive and concurrent systems.  Practical model checking tools like SPIN, NuSMV, and many others (\cite{Holzmann:1997:SPIN,NuSMV:2000:Cimatti,Cook+Koskinen+Vardi/2011/TemporalPropertyVerificationAsAProgramAnalysisTask}) automatically verify whether a given computer system, such as a hardware circuit or a communication protocol, is correct with respect to its \LTL specification. The basic principle, as introduced in 1977 by Amir Pnueli~\cite{Pnueli/1977/TheTemporalLogicOfPrograms}, is to specify the correctness of a program as a set of infinite sequences, called \emph{traces}, which define the acceptable executions of the system.

Hyperproperties, i.e., properties which relate multiple execution traces, cannot be specified in \LTL. Such properties are of prime interest in information flow security, where dependencies between the secret inputs and the publicly observable outputs of a system are considered potential security violations.
Commonly known properties of that type are noninterference~\cite{DBLP:conf/sp/Roscoe95,DBLP:journals/jcs/McLean92} or observational determinism~\cite{DBLP:conf/csfw/ZdancewicM03}.
In other settings, relations between traces are explicitly desirable: robustness properties, for example, state that similar inputs lead to similar outputs.
Hyperproperties are not limited to the area of information flow control.
E.g., distributivity and other system properties like fault tolerance can be expressed as hyperproperties~\cite{DBLP:journals/acta/FinkbeinerHLST20}.

The main approach to specify hyperproperties has been to extend temporal logics like \LTL, \CTL, and \QPTL with explicit trace and path quantification, resulting in logics like \hyltl~\cite{DBLP:conf/post/ClarksonFKMRS14}, \hyctl~\cite{DBLP:conf/post/ClarksonFKMRS14}, and \HQPTL~\cite{MarkusThesis,DBLP:conf/lics/CoenenFHH19}. 
Most frequently used is \hyltl, which can express noninterference as follows:
	\(
	\forall \pi \ldot \forall \pi' \ldot \G (\bigwedge_{i \in \mathit{I}} i_\pi \leftrightarrow i_{\pi'}) \rightarrow \G (\bigwedge_{o \in \mathit{O}} o_\pi \leftrightarrow o_{\pi'})
	\).
%
The formula states that any two traces which globally agree on the value of the public inputs $I$ also globally agree on the public outputs $O$. Consequently, the value of secret inputs cannot affect the value of the publicly observable outputs.

It is not clear, however, whether quantification over traces is the best way to express hyperproperties. The success of LTL over first-order logics for the specification of linear-time properties stems from the fact that its modal operators replace explicit quantification of points in time. This allows for a much more concise and readable formulation of the same property.
The natural question to ask is whether a purely modal logic for hyperproperties would have similar advantages.
A candidate for such a logic is \LTL with \emph{team semantics} ~\cite{kmvz18}. Under team semantics, \LTL expresses hyperproperties without explicit references to traces. Instead, each subformula is evaluated with respect to a set of traces, called a \emph{team}. Temporal operators advance time on all traces of the current team. Using the split operator $\lor$, teams can be split during the evaluation of a formula, which enables us to express properties of subsets of traces.

As an example, consider the property that there is a point in time, common for all traces, after which a certain event $a$ does not occur any more. We need a propositional and a trace quantifier to express such a property in HyperQPTL (it is not expressible in HyperLTL).
The formula $\exists p \ldot \forall \pi \ldot \LTLeventually p \land \LTLglobally (p \rightarrow \LTLglobally \neg a_\pi)$ states that there is a $p$-sequence $s \in (\pow{\{p\}})^\omega$ such that $p$ is set at least once, and if $p \in s[i]$, then $a$ is not set on all traces $\pi$ on all points in time starting from $i$. The same property can be expressed in \teamltl without any quantification simply as $\LTLeventually \LTLglobally \neg a$. The formula exploits the synchronous semantics of \teamltl by stating that there is a point such that for all future points all traces have $a$ not set.
As a second example, consider the case that an unknown input determines the behaviour of the system. Depending on the input, its execution traces either agree on $a$ or on $b$.
We can express the property in HyperLTL with three trace quantifiers:
	\(
	\exists \pi_1 \ldot \exists \pi_2 \ldot \forall \pi \ldot \LTLglobally(a_{\pi_1} \leftrightarrow a_\pi) \lor \LTLglobally(b_{\pi_2} \leftrightarrow b_\pi)
	\).
In \teamltl, the same property can be simply expressed as $\LTLglobally(a \clor \neg a) \lor \LTLglobally(b \clor \neg b)$. The Boolean or operator $\clor$ expresses that in the current team, either the left side holds on all traces or the right side does.

The use of the $\clor$ operator reveals another strength of \teamltl: its modularity. The research on team semantics (see related work section) has a rich tradition of studying extensions of team logics with new atomic statements and operators. They constitute a well-defined way to increase a logic's expressiveness in a step-by-step manner.
Besides $\clor$, examples are Boolean negation $\cneg$, the inclusion atom $\subseteq$, and universal subteam quantifiers $\asub$ and $\flatop$. Inclusion atoms have been found to be fascinating for their ability to express recursion in the first-order setting; the expressivity of $\mathrm{FO}(\subseteq)$ coincides with greatest fixed point logic and hence PTIME \cite{gallhella13}. In turn, all $\LTL$-definable properties can be expressed by \teamltl-formulae of the form $\flatop \varphi$. With the introduction of generalised atoms, \teamltl even permits custom extensions.
Possibly most interesting in the context of hyperproperties are dependence atoms. A dependence atom $\dep(x_1,\ldots,x_n)$ is satisfied by a team $X$ if any two assignments assigning the same values to the variables $x_1,\ldots, x_{n-1}$ also assign the same value to $x_n$. For example, the \teamltl formula 
	\(
	(\G \dep(i_1, i_2, o)) \lor (\G \dep(i_2, i_3, o))
	\)
states that the executions of the system can be decomposed into two parts; in the first part, the output $o$ is determined by the inputs $i_1$ and $i_2$, and in the second part, $o$ is determined by the inputs $i_2$ and $i_3$. 

Temporal team logics constitute a new, fundamentally different approach to specify hyperproperties. While \hyltl and other quantification-based hyperlogics have been studied extensively (see section on related work), 
only few results are known about the expressive power and complexity of \teamltl and its variants. In particular, we know very little about how the expressivity of the two approaches compares.
What is known is that \hyltl and \teamltl are incomparable in expressivity~\cite{kmvz18} and that the model checking problem of \teamltl without splitjunctions $\lor$ (what makes the logic significantly weaker) is in $\PSPACE$~\cite{kmvz18}.
On the other hand, it was  recently  shown that  the complexity of satisfiability and model checking of \teamltl with Boolean negation $\cneg$ is equivalent to the decision problem of third-order arithmetic~\cite{LUCK2020} and hence highly undecidable.

\noindent{\bf Our contribution.}
We advance the understanding of team-based logics for hyperproperties by exploring the relative expressivity of \teamltl and temporal hyperlogics like \hyltl, as well as the decidability frontier of the model checking problem of \teamltl.
Our expressivity and model checking results are summarized in Table~\ref{tab:EXPResults} and Table~\ref{tab:MCResults}.
We identify expressively complete extensions of \teamltl (displayed on the left of Table~\ref{tab:EXPResults}) that can express all (all downward closed, resp.) Boolean relations on $\LTL$-properties of teams, and present several translations from team logics to hyperlogics.	
We begin by approaching the decidability frontier of \teamltl from above, and tackle a question posed in \cite{LUCK2020}: \emph{Does some sensible restriction to the use of Boolean negation in $\teamltl(\cneg)$ yield a decidable logic?} We show that already
a very restricted access to $\cneg$ leads to high undecidability, whereas already the use of inclusion atoms $\subseteq$ together with Boolean disjunctions $\clor$ suffices for undecidable model checking. 
Furthermore, we establish that these complexity results transfer to the satisfiability problem of the related logics. 
Next, regarding the expressivity of \teamltl, we show that
its extensions with all (all downward closed, resp.) atomic $\LTL$-properties of teams translate to simple fragments of \HQPTLP. 
Consequently, known decidability results for quantification-based hyperlogics enable us to approach the decidability frontier of \teamltl extensions from below.
We establish an efficient translation from the so-called \emph{$k$-coherent fragment} of $\teamltl(\cneg)$ to  the universal fragment of \hyltl (for which model checking is $\PSPACE$-complete \cite{conf/cav/FinkbeinerRS15}) and thereby obtain $\EXPSPACE$ model checking for the fragment. 
Finally, we show that the so-called \emph{left-flat}  fragment of $\teamltl(\clor,\flatop)$  enjoys decidable model checking via  a translation to $\existsu_p^*\forall_\pi^*\HQPTL$. 
\begin{table}[t]
	\setlength{\tabcolsep}{7pt}
	\centering
	\begin{tabular}{ccccc}	
		&  && (assuming left-flatness) &\\
		& $\teamltl(\clor, \flatop)$ & $\stackrel{\mathrm{Thm. \ref{thm:translation}}}{\leq}$ & $\existsu_q^*\forall_\pi \HQPTL$ &\\
		&  &$\stackrel{ \mathrm{Thm. \ref{s:teamltl2hyperqptl}}}{\leq}$& $\existsso_p \quantu^*_p \forall_\pi \HQPTLP$ & \\
		&  \rotatebox[origin=c]{-90}{$<$}${}^\dagger$ && & \\
		& $\teamltl(\clor, {\cneg\!\bot}, \flatop)$ &$\stackrel{\mathrm{Thm. \ref{s:teamltl2hyperqptl}}}{\leq}$& $\existsso_p \quantu^*_p \exists_\pi^*\forall_\pi \HQPTLP$& \\
		&&&&\\
		&\hspace{3mm}  \rotatebox[origin=c]{-90}{$\leq$}  {\cite{LUCK2020}} && (assuming $k$-coherence) & \\
		& $\teamltl(\cneg)$ &$\stackrel{\mathrm{Thm. \ref{thm:kcoherent-team-hyper}}}{\leq}$&  $\forall^k\hyltl$ & 
	\end{tabular}
	\vspace{0.5em}
	\caption{Expressivity results. The logics $\teamltl(\clor, \cneg\!\bot, \flatop)$  and $\teamltl(\flatop, \clor)$ can express \emph{all}/\emph{all downward closed} atomic $\LTL$-properties of teams (see the discussion at the end of Section \ref{sec:teamltl}).
		$\dagger$ holds since $\teamltl(\flatop,\clor)$ is downward closed.\looseness=-1}
	\label{tab:EXPResults}
	\vspace{-1em}
\end{table}
\begin{table}[t]
	\setlength{\tabcolsep}{7pt}
	\def\arraystretch{1.2}
	\centering
	\begin{tabular}[t]{l|l}
		Logic & Model Checking Result \\ \hline \hline
		$\teamltl$ without $\lor$ & in $\PSPACE$~\cite{kmvz18} \\
		$k$-coherent $\teamltl(\cneg)$ & in $\EXPSPACE$~[Thm.~\ref{cor:kcoherentMC}] \\
		left-flat $\teamltl(\clor,\flatop)$ & in $\EXPSPACE$~[Thm.~\ref{thm:leftflatMC}] \\			
		$\teamltl(\subseteq, \clor)$ & $\Sigma^0_1$-hard [Thm.~\ref{thm:undecidable}]\\
		$\teamltl(\subseteq, \clor, \asub)$ & $\Sigma^1_1$-hard [Thm.~\ref{thm:sig11_unsat}] \\
		$\teamltl(\cneg)$ & \makecell[l]{complete for third-order arithmetic~\cite{LUCK2020}}
	\end{tabular}
	\vspace{0.5em}
	\caption{Complexity results.}
	\label{tab:MCResults}
	\vspace{-1em}
\end{table}

\noindent{\bf Related work.}
The development of team semantics began with the introduction of Dependence Logic~\cite{vaananen07}, which adds the concept of functional dependence to first-order logic by means of new atomic dependence formulae.
During the past decade, team semantics has been generalised to propositional \cite{YANG2017}, modal \cite{vaananen08}, temporal \cite{KrebsMV15}, and probabilistic \cite{HKMV18} frameworks, and fascinating connections to fields such as database theory~\cite{HannulaK16}, statistics \cite{CoranderHKPV16}, real valued computation \cite{HannulaKBV20}, and quantum information theory \cite{Hyttinen15b} has been identified.	
In the modal team semantics setting, model checking and satisfiability problems have been shown to be decidable, see \cite[page 627]{HellaKMV19} for an overview of the complexity landscape. Expressivity and definability of related logics is also well understood, see, e.g. \cite{HellaLSV14,KontinenMSV15,SanoV19}. The study of temporal logics with team semantics, was initiated in \cite{KrebsMV15}, where team semantics for computational tree logic \CTL was given. The idea to develop team-based logics for hyperproperties was coined in \cite{kmvz18}, where \teamltl was first introduced and shown incomparable to \hyltl.
The interest on logics for hyperproperties, so-called hyperlogics, was sparked by the introduction of \hyltl and \hyctl~\cite{DBLP:conf/post/ClarksonFKMRS14}.
Many temporal logics have since been extended with trace and path quantification to obtain various hyperlogics, e.g., to express asynchronous hyperproperties~\cite{DBLP:journals/pacmpl/GutsfeldMO21, DBLP:conf/cav/BaumeisterCBFS21}, hyperproperties on finite traces~\cite{DBLP:conf/ijcai/GiacomoFMP21}, probabilistic hyperproperties~\cite{DBLP:conf/qest/AbrahamB18}, or timed hyperproperties~\cite{DBLP:conf/time/HoZ019}.
Model checking \hyltl and the strictly more expressive \HQPTL is decidable, though $k$-$\EXPSPACE$-complete, where $k$ is the number of quantifier alternations in the formula~\cite{conf/cav/FinkbeinerRS15, MarkusThesis}.
Model checking \HQPTLP, on the other hand, is undecidable~\cite{DBLP:conf/cav/FinkbeinerHHT20}.
The expressivity of \hyltl, \hyctl, and \HQPTL has been compared to first-order and second-order hyperlogics resulting in a hierarchy of hyperlogics~\cite{DBLP:conf/lics/CoenenFHH19}.
Beyond model checking and expressivity questions, especially \hyltl has been studied extensively.
This includes its satisfiability~\cite{DBLP:conf/concur/FinkbeinerH16,DBLP:conf/csl/Mascle020}, runtime monitoring ~\cite{DBLP:journals/fmsd/FinkbeinerHST19,DBLP:conf/csfw/AgrawalB16} and enforcement problems~\cite{RuntimeEnforcementHyperLTL}, as well as synthesis~\cite{DBLP:journals/acta/FinkbeinerHLST20}. 
%

\section{Basics of TeamLTL}\label{sec:teamltl}

Let us start by recalling the syntax of \LTL from the literature.
Fix a set $\ap$ of \emph{atomic propositions}. 
The set of formulae of \LTL (over $\ap$) is generated by the following grammar:
\[
\varphi \ddfn p \mid \neg p  \mid \varphi \lor \varphi \mid \varphi\land \varphi \mid \X \varphi \mid \varphi \U \varphi \mid \varphi \W \varphi, \quad\quad  \text{where $p \in \ap$.}
\]
We adopt, as is common  in studies on team logics, the convention that formulae are given in negation normal form. The logical constants $\top,\bot$ and connectives $\rightarrow,\leftrightarrow$ are defined as usual (e.g., $\bot \dfn p \land \neg p$ and $\top \dfn p \lor \neg p$), and $\LTLeventually \phi:=\top\U\phi$ and $\G \phi:=\phi\W\bot$.

A \emph{trace} $t$ over $\ap$ is an infinite sequence from $(\pow{\ap})^\omega$. For a natural number $i\in\N$, we denote by $t[i]$ the $i$th element of $t$ and by $t[i,\infty]$ the postfix $(t[j])_{j\geq i}$ of $t$. 
The satisfaction relation $(t,i)\models \varphi$, for  \LTL formulae $\phi$, is defined as usual, see e.g., \cite{PitermanP18}. We use $\llbracket \varphi \rrbracket_{(t,i)}\in\{0,1\}$ to denote the truth value of $\varphi$ on $(t,i)$.
%
A {\em (temporal) team} is a pair $(T,i)$ consisting a set of traces $T\subseteq (\pow{\ap})^\omega$ and a natural number $i\in\N$ representing the time step.
We write $T[i]$ and $T[i,\infty]$ to denote the sets $\{t[i]\mid t\in T \}$ and $\{t[i,\infty]\mid t\in T \}$, respectively.

Let us next introduce the logic $\LTL$ interpreted with team semantics (denoted $\teamltl$).
$\teamltl$ was first studied in \cite{kmvz18}, where it was called $\LTL$ with {\em synchronous} team semantics.  The satisfaction relation $(T,i)\models\phi$ for $\teamltl$  is defined as follows:
\begin{align*}
&(T, i)\models p   &&\text{iff}&& \forall t \in T: p\in t[i]  \\
&(T, i)\models \lnot p  &&\text{iff}&& \forall t \in T: p\notin t[i]\\
&(T, i)\models \phi\land\psi  &&\text{iff}&& (T,i)\models\phi  \text{ and } (T,i)\models\psi \\
&(T,i) \models\X\varphi  &&\text{iff}&& (T, i+1)\models\varphi \\
&(T, i)\models \phi\land\psi  &&\text{iff}&& (T,i)\models\phi  \text{ and } (T,i)\models\psi \\
&(T,i)  \models\X\varphi  &&\text{iff}&& (T, i+1)\models\varphi \\
&(T, i)\models \phi\lor\psi  &&\text{iff}&& (T_1,i)\models\phi \text{ and }(T_2, i)\models\psi, \text{ for some } T_1,T_2\text{ s.t. }T_1\cup T_2=T\\
&(T,i)\models\phi\U\psi &&\text{iff}&& \exists k\ge i\text{ such that } (T, k)\models\psi \text{ and }
\forall m: i \leq m <k \Rightarrow (T, m)\models\phi\\
&(T,i)\models\phi\W\psi &&\text{iff}&&\forall k \ge i : (T,k)\models\phi \text{ or }
\exists m \text{ such that }
i\le m\leq k \text{ and } (T,m)\models\psi 
\end{align*}
Note that $(T,i)\models \bot$ iff $T=\emptyset$.
Two formulae $\phi$ and $\psi$ are {\em equivalent} (written $\phi\equiv\psi$), if the equivalence $(T,i)\models\phi$ iff $(T,i)\models\psi$ holds for every $(T,i)$. 
We say that a logic $\mathcal{L}_2$ is {\em at least as expressive} as a logic $\mathcal{L}_1$ (written $\mathcal{L}_1\leq\mathcal{L}_2$) if for every $\mathcal{L}_1$-formula $\phi$, there exists an $\mathcal{L}_2$-formula $\psi$ such that $\phi\equiv\psi$. We write $\mathcal{L}_1\equiv\mathcal{L}_2$ if both $\mathcal{L}_1\leq\mathcal{L}_2$ and $\mathcal{L}_2\leq\mathcal{L}_1$ hold.
The following are important semantic properties of formulae from the team semantics literature:

\begin{description}
\item[(Downward closure)]
If $(T,i)\models\phi$ and $S\subseteq T$, then $(S,i)\models\phi$.
\item[(Empty team property)] $(\emptyset,i)\models\phi$.
\item[(Flatness)] $(T,i)\models\phi$ ~~iff~~ $(\{t\},i)\models\phi$ for all $t\in T$.
\item[(Singleton equivalence)] $(\{t\},i)\models\phi$ ~~iff~~ $(t,i)\models\phi$.
\end{description}

\noindent A logic has one of the above properties if every formula of the logic has the property.
\teamltl satisfies {downward closure}, singleton equivalence,
and the empty team property~\cite{kmvz18}. 
However, it
does not
satisfy flatness; for instance, the formula $\LTLeventually p$ is not flat.

The power of team semantics comes with the ability to enrich logics with novel atomic statements describing properties of teams. We thereby easily get a hierarchy of team logics of different expressiveness.
The most prominent examples of such atoms are {\em dependence atoms} $\dep(\phi_1,\dots,\phi_n,\psi)$ and {\em inclusion atoms} $\phi_1,\dots,\phi_n  \subseteq  \psi_1,\dots,\psi_n$, with $\phi_1,\dots,\phi_n ,\psi,\psi_1,\dots,\psi_n$ being $\LTL$-formulae. The dependence atom states that the truth value of $\psi$ is functionally determined by that of $\phi_1, \ldots, \phi_n$. The inclusion atom states that each value combination of $\phi_1, \ldots, \phi_n$ must also occur as a value combination for $\psi_1, \ldots, \psi_n$. Their formal semantics is defined as:
\[
(T, i)\models  \dep(\phi_1,\dots,\phi_n,\psi)  ~~\text{iff}~~ \forall t ,t' \in T :
\Big( \bigwedge_{1\leq j\leq n}  \truth{\phi_j}_{(t,i)} = \truth{\phi_j}_{(t',i)} \Big) \Rightarrow \truth{\psi}_{(t,i)} = \truth{\psi}_{(t',i)}
\]
%
\[
(T, i)\models \phi_1,\dots,\phi_n  \subseteq  \psi_1,\dots,\psi_n  ~~\text{iff}~~ \forall t\in T \, \exists t'\in T :
\bigwedge_{1\leq j\leq n}  \truth{\phi_j}_{(t,i)} = \truth{\psi_j}_{(t',i)}		
\]

As an example, let $o_1, \ldots, o_n$ be some observable outputs and $s$ be a secret. The atom $(o_1, \ldots, o_n, s) \subseteq (o_1, \ldots, o_n, \neg s)$ expresses a form of \emph{non-inference} by stating that an observer cannot infer the current value of the secret from the outputs.
We also consider other connectives known in the team semantics literature: {\em Boolean disjunction} $\clor$, {\em Boolean negation} $\cneg$, and  {\em universal subteam quantifiers} $\asub$ and $\flatop$, with their semantics defined as:
\[
\begin{array}{lcllcl}
(T,i) \models \varphi \clor \psi  &\text{iff}& (T,i) \models \varphi \text{ or } (T,i) \models \psi \\
(T,i) \models \,\cneg \varphi  &\text{iff}& (T,i) \not\models \varphi \\
(T,i) \models \asub \varphi  &\text{iff}& \forall S\subseteq T: (S,i) \models \varphi \\
(T,i) \models \flatop \varphi  &\text{iff}& \forall t\in T: (\{t\},i) \models \varphi
\end{array}
\]
If $\mathcal A$ is a collection of atoms and connectives, we let $\teamltl(\mathcal A)$ denote the extension of $\teamltl$ with the atoms and connectives in $\mathcal A$. For any atom or connective $\circ$, we write simply $\teamltl(\mathcal A, \circ)$ instead of $\teamltl(\mathcal A \cup \{\circ\})$.

$\teamltl(\cneg)$ is a very expressive logic; all of the above connectives and atoms, as well as many others, have been shown to be definable in ${\teamltl(\cneg)}$~\cite{HannulaKVV18,LUCK2020}.
%
To systematically explore less expressive variants of \teamltl, we introduce two representative logics of different expressiveness, namely $\teamltl(\clor, \flatop)$ and $\teamltl(\clor, {\cneg\!\bot}, \flatop)$. The expression ${\cneg\!\bot}$ can be used to enforce non-emptiness of a team. What makes these logics good representatives is their semantic property that they can express a general class of Boolean relations.
Let $B$ be a set of $n$-ary Boolean relations. We define the property $[\varphi_1,\dots,\varphi_n]_B$ for an $n$-tuple $(\varphi_1,\dots,\varphi_n)$ of $\LTL$-formulae:
$$
(T,i) \models [\varphi_1,\dots,\varphi_n]_B \quad\text{iff}\quad \{ (\truth{\phi_1}_{(t,i)}, \dots,  \truth{\phi_n}_{(t,i)} ) \mid t\in T  \} \in B.
$$	

The logic $\teamltl(\clor, {\cneg\!\bot}, \flatop)$ is expressively complete with respect to all $[\varphi_1,\dots,\varphi_n]_B$. 
That is, for every set of Boolean relations $B$ and $\LTL$-formulae $\varphi_1,\dots,\varphi_n$, the property $[\varphi_1,\dots,\varphi_n]_B$ is expressible in $\teamltl(\clor, {\cneg\!\bot}, \flatop)$.
Furthermore, $\teamltl(\clor, \flatop)$ can express all downward closed ($S_1 \in B$ \& $S_2 \subseteq S_1$ imply $S_2 \in B$) $B$.
These results are reformulated and proved using so-called \emph{generalised atoms} in \ifbool{appdx}{Appendix \ref{a:sec:GA}}{the extended version of this paper~\cite{arxivVersion}}. 
Note that, e.g., $k$-ary inclusion and dependence atoms can be defined using suitable Boolean relations $B$. Indeed it follows that, from the expressivity point-of-view, $\teamltl(\clor, \flatop)$ and $\teamltl(\clor, {\cneg\!\bot}, \flatop)$ subsume all extensions of $\teamltl$ with downward closed (resp. all) \emph{atomic notions of dependence}, i.e., atoms which state some sort of functional (in)dependence, like the dependence atom (which is downward closed) or the inclusion atom (which is not).

\section{Undecidable Extensions of TeamLTL}\label{sec:undec}

In \cite{LUCK2020}, L\"uck established that the model checking problem for $\teamltl(\cneg)$ is highly undecidable.
The proof heavily utilises the interplay between Boolean negation $\cneg$ and disjunction $\lor$; it was left as an open problem whether some sensible restrictions on the use of the Boolean negation would lead toward discovering decidable logics.
We show that, on the contrary, the decidability bounds are much tighter.
Already $\teamltl(\subseteq,\clor)$ (which is subsumed by $\teamltl(\clor, {\cneg\!\bot}, \flatop)$) is undecidable, and already very restricted access to $\cneg$ (namely, a single use of the $\asub$ quantifier) leads to high undecidability.

We define the model checking problem based on Kripke structures $\kK=(W, R, \eta , w_0)$, where $W$ is a finite set of states, $R\subseteq W^2$ the transition relation, $\eta\colon W\rightarrow 2^\ap$ a labelling function, and $w_0\in W$ an initial state of $W$.
A path $\sigma$ through $\kK$ is an infinite sequence $\sigma \in W^\omega$ such that $\sigma[0]= w_0$ and $(\sigma[i], \sigma[i + 1]) \in R$ for every $i \geq 0$. The trace of $\sigma$ is defined as $t(\sigma) \dfn \eta(\sigma[0])\eta(\sigma[1])\dots \in (2^\ap)^\omega$. A Kripke structure $\kK$ induces a set of traces $\traces(\kK) = \{t(\sigma) \mid \sigma \text{ is a path through $\kK$}\}$.

\begin{definition}
The \emph{model checking problem of a logic $\LL$} is the following decision problem:
Given a formula $\varphi\in\LL$  and a Kripke structure $K$ over $\ap$, determine whether $(\mathit{Traces}(K), 0) \models \varphi$.
\end{definition}

Our undecidability results are obtained by reductions from \emph{non-deterministic 3-counter machines}.	
A non-deterministic 3-counter machine $M$ consists of a list $I$ of $n$ instructions that manipulate three counters $C_l$, $C_m$, and $C_r$. All instructions are of the following forms:

\begin{itemize}
\item $C_a^+ \text{ goto } \{j_1, j_2\}$, \qquad $C_a^- \text{ goto } \{j_1, j_2\}$, \qquad $\text{if } C_a=0 \text{ goto }  j_1 \text{else goto }  j_2$,
\end{itemize}

\noindent where $a\in\{l,m,r\}$, $0\leq j_1,j_2 < n$. A \emph{configuration} is a tuple $(i,j,k,t)$, where $0\leq i < n$ is the next instruction to be executed, and $j,k,t\in \N$ are the current values of the counters $C_l$, $C_m$, and $C_r$. The execution of the instruction  $i\colon C_a^+ \text{ goto } \{j_1, j_2\}$ ($i\colon C_a^- \text{ goto } \{j_1, j_2\}$, resp.) increments (decrements, resp.) the value of the counter $C_a$ by $1$. The next instruction is selected nondeterministically from the set $\{j_1, j_2\}$. The instruction  $i\colon \text{if } C_a=0 \text{ goto }  j_1, \text{else goto }  j_2$ checks whether the value of the counter $C_a$ is currently $0$ and proceeds to the next instruction accordingly. The \emph{consecution relation} $\cons$ of configurations is defined as usual.  The \emph{lossy consecution relation} $(i_1,i_2,i_3,i_4)\lcons (j_1,j_2,j_3,j_4)$ of configurations holds if $(i_1,i'_2,i'_3,i'_4)\cons (j_1,j'_2,j'_3,j'_4)$ holds for some $i'_2, i'_3,i'_4, j'_2, j'_3,j'_4$ with  $i_2 \geq i'_2$, $i_3\geq i'_3$, $i_4\geq i'_4$, $j'_2 \geq j_2$, $j'_3\geq j_3$, and $j'_4\geq j_4$. A \emph{(lossy) computation} is an infinite sequence of (lossy) consecutive configurations starting from the initial configuration $(0,0,0,0)$. A (lossy) computation is \emph{$b$-recurring} if the instruction labelled $b$ occurs infinitely often in it.
%
Deciding whether a given non-deterministic $3$-counter machine has a $b$-recurring ($b$-recurring lossy) computation for a given $b$ is $\Sigma_1^1$-complete ($\Sigma_1^0$-complete, resp.)~\cite{AlurH94, Schnoebelen10}.


We reduce the existence of a $b$-recurring lossy computation of a given $3$-counter machine $M$ and an instruction label $b$ to the model checking problem of $\teamltl(\subseteq, \clor)$. 
We also illustrate that with a single instance of $\asub$ we can enforce non-lossy computation instead.

\begin{theorem}\label{thm:undecidable}
Model checking for $\teamltl(\subseteq, \clor)$ is $\Sigma_1^0$-hard.
\end{theorem}
\begin{proof}
Given a set $I$ of instructions of a $3$-counter machine $M$, and an instruction label $b$, we construct a $\teamltl(\subseteq, \clor)$-formula $\varphi_{I,b}$ and a Kripke structure $\kK_I$ such that
\begin{equation}\label{thm:undecidable_eq1}
\big(\traces(\kK_I),0\big) \models \varphi_{I,b} \quad\text{iff}\quad\text{$M$ has a $b$-recurring lossy computation.} 
\end{equation}
The $\Sigma_1^0$-hardness then follows since our construction is clearly computable.		
The idea is the following:
Put $n\dfn \lvert I \rvert$. A set $T$ of traces using propositions $\{c_l,c_m,c_r,d,0,\dots, n-1\}$ encodes the sequence $(\vec{c}_j)_{j\in\N}$  of configurations, if for each $j\in\N$ and $\vec{c}_j=(i,v_l,v_m,v_r)$ 

\begin{itemize}
\item $t[j]\cap\{0,\dots, n-1\} =\{i\}$, for all $t\in T$,
\item $\lvert \{t[j,\infty] \mid c_s\in t[j], t\in T \} \rvert = v_s$, for each $s\in\{l,m,r\}$.
\end{itemize}

\noindent Hence, we use $T[j,\infty]$ to encode the configuration $\vec{c}_j$; the propositions $0,\dots, n-1$ are used to encode the next instruction, and $c_l,c_m,c_r,d$ are used to encode the values of the counters. The proposition $d$ is a dummy proposition used to separate traces with identical postfixes with respect to $c_l$, $c_m$, and $c_r$.
The Kripke structure $\kK_I = (W,R,\eta,w_0)$ over the set of propositions $\{c_l,c_m,c_r,d,0,\dots, n-1\}$ is defined such that every possible sequence of configurations of $M$ starting from $(0,0,0,0)$ can be encoded by some team $(T,0)$, where $T\subseteq \traces(\kK_I)$.
A detailed construction of the formula $\varphi_{I,b}$ and the Kripke structure $\kK_I$ together with a detailed proof for the fact that (\ref{thm:undecidable_eq1}) indeed holds can be found in~\ifbool{appdx}{Appendix \ref{a:sec:undec}}{the full version of this paper \cite{arxivVersion}}.	
\end{proof}

The underlying reason for utilising lossy computations in the above proof is the following: In our encoding, we use the cardinality of the set $\lvert \{t[j,\infty] \mid c_l\in t[j], t\in T \} \rvert$ to encode the value of the counter $C_l$ in the $j$th configuration. It might, however, happen that two distinct traces $t,t'\in T$ have the same postfix, that is, $t[j,\infty]=t'[j,\infty]$, for some $j\in\N$. The collapse of two traces encoding distinct increments of the counter $C_l$ then corresponds to the uncontrollable decrement of the counter values in lossy computations. Using the universal team quantifier $\asub$ we can forbid this effect, and encode non-lossy computations.
The proof of the following theorem can be found in~\ifbool{appdx}{Appendix \ref{a:sec:undec}}{the full version of this paper \cite{arxivVersion}}.

\begin{restatable}{theorem}{sigunsat}\label{thm:sig11_unsat}
Model checking for $\teamltl(\subseteq, \clor, \asub)$ is $\Sigma_1^1$-hard. This holds already for the fragment with a single occurrence of $\asub$.
\end{restatable}

As is common in the $\LTL$-setting, the model checking problem of $\teamltl(\subseteq, \clor)$ can be embedded in its satisfiability problem using auxiliary propositions and  $\teamltl(\subseteq, \clor)$-formulae.
A formula $\varphi$ is satisfiable, if there exists a non-empty $T$ such that $(T,0)\models \varphi$.
We thus obtain the following corollary, which is also proven in~\ifbool{appdx}{Appendix \ref{a:sec:undec}}{the full version of this paper \cite{arxivVersion}}.

%

\begin{restatable}{corollary}{corundecsat}\label{cor:undecsat}
The satisfiability problems for $\teamltl(\subseteq, \clor)$ and $\teamltl(\subseteq, \clor, \asub)$ are $\Sigma_1^0$-hard and $\Sigma_1^1$-hard, resp.
\end{restatable}

\section{Quantification-based Hyperlogics and Team Semantics}

In this section, we define those quantification-based hyperlogics against which we compare \teamltl in the rest of the paper.
$\teamltl$ and $\hyltl$ are known to have orthogonal expressivity~\cite{kmvz18} but apart from that, nothing is known about the relationship between the different variants of \teamltl and other temporal hyperlogics such as \HQPTL~\cite{MarkusThesis,DBLP:conf/lics/CoenenFHH19}. We aim to identify fragments of the logics with similar expressivity to better understand the relative expressivity of \teamltl for the specification of hyperproperties.




\HQPTLP~\cite{DBLP:conf/cav/FinkbeinerHHT20} is a temporal logic for hyperproperties. 
It subsumes \hyltl and \HQPTL, so we proceed to give a definition of \HQPTLP and define the latter logics as fragments.
\HQPTLP extends \LTL with explicit trace quantification and quantification of atomic propositions. 
As such, it also subsumes \QPTL, which can express all $\omega$-regular properties.
Fix an infinite set $\pathvars$ 
of trace variables.
$\HQPTLP$ has three types of quantifiers, one for traces and two for propositional quantification.
\begin{align*}
\varphi &\ddfn  \forall\pi\ldot\varphi \mid \exists\pi\ldot\varphi \mid \forallu p\ldot\varphi \mid \existsu p \ldot\varphi \mid \forallso p\ldot\varphi \mid \existsso p \ldot\varphi \mid \psi \enspace \\
\psi &\ddfn p_\pi \mid \neg p_\pi \mid \psi\lor\psi \mid \psi\land\psi \mid \X\psi \mid \psi\U\psi \mid \psi\W\psi
\end{align*}
Here, $p \in \ap$, $\pi \in \pathvars$, and $\forall\pi$ and $\exists\pi$ stand for universal and existential trace quantifiers, $\forallso p$ and $\existsso p$ stand for (non-uniform) propositional quantifiers, and $\forallu p$ and $\existsu p$ stand for uniform propositional quantifiers.
We also study two syntactic fragments of  $\HQPTLP$. $\HQPTL$ is $\HQPTLP$ without non-uniform propositional quantifiers, and $\hyltl$ is $\HQPTLP$ without any  propositional quantifiers. In the context of $\HQPTL$, 
we also 
write $\forall p$ and $\exists p$  instead of $\forallu p$ and $\existsu p$.
For an $\LTL$-formula $\varphi$ and trace variable $\pi$, we let $\varphi_\pi$ denote the $\hyltl$-formula obtained from $\varphi$ by replacing all proposition symbols $p$ by their indexed versions $p_\pi$. We extend this convention to tuples of formulae as well.		

The semantics of $\HQPTLP$  is defined over a set $T$ of traces. Intuitively, the atomic formula $p_\pi$ asserts that  $p$ holds on trace $\pi$. Uniform propositional quantifications $\forallu p$ and $\existsu p$ add an atomic proposition $p$ such that all traces agree on the valuation of $p$ on any given time step $i$, whereas non-uniform propositional quantifications $\forall p$ and $\exists p$ colour the traces in $T$ in an arbitrary manner.
Non-uniform propositional quantification thus implements true second-order quantification, whereas uniform propositional quantification can be interpreted as a quantification of a set of points in time.

A \emph{trace assignment} is a function $\pathassign : \pathvars \to T$ that maps each trace variable in $\pathvars$ to some trace in $T$.
%
A \emph{modified trace assignment} $\pathassign[\pi \mapsto t]$ is equal to $\pathassign$  except that $\Pi[\pi \mapsto t](\pi)=t$.
For any subset $A\subseteq\ap$, we write $t\upharpoonright A$ for the projection of $t$ on $A$ (i.e., $(t\upharpoonright A)[i] \dfn t[i]\cap A$ for all $i \in \N$).
For any two trace assignments $\pathassign$ and $\pathassign'$, we write $\pathassign =_A \pathassign'$, if $\big(\pathassign(\pi)\upharpoonright A\big) = \big(\pathassign'(\pi)\upharpoonright A\big)$ for all $\pi\in\pathvars$. Similarly, $T =_A T'$ whenever $ \{ t\upharpoonright A \mid t\in T\} =   \{ t\upharpoonright A \mid t\in T'\}$. 
For a sequence $s \in (\pow{\set{p}})^\omega$ over a single propositional variable $p$, we write $T[p \mapsto s]$ for the set of traces obtained from $T$ by reinterpreting $p$ on all traces as in $s$ while ensuring that $T[p \mapsto s] =_{\ap\setminus\{p\}} T$.
We use $\pathassign [p \mapsto s]$ accordingly.
The satisfaction relation $\Pi,i\models_T\phi$ for $\HQPTLP$-formulae $\phi$  is defined as follows:
\begin{alignat*}{3}
&\pathassign,i \models_T p_\pi       \quad&&\text{iff} \quad && p \in \pathassign(\pi)[i] \\
&\pathassign,i \models_T \varphi_1 \lor \varphi_2         \quad && \text{iff} && \pathassign,i \models_T \varphi_1 \text{ or } \pathassign,i \models_T \varphi_2\\			
&\pathassign,i \models_T \neg p_\pi           \quad&& \text{iff} && p \not\in \pathassign(\pi)[i] \\
&\pathassign,i \models_T \varphi_1 \land \varphi_2         \quad && \text{iff}  && \pathassign,i \models_T \varphi_1 \text{ and } \pathassign,i \models_T \varphi_2 \\
&\pathassign,i \models_T p_\pi       &&\text{iff} && p \in \pathassign(\pi)[i] \\
&\pathassign,i \models_T \neg p_\pi           && \text{iff} &&p \not\in \pathassign(\pi)[i] \\
&\pathassign,i \models_T \varphi_1 \lor \varphi_2         \quad && \text{iff} \quad&&\pathassign,i \models_T \varphi_1 \text{ or } \pathassign,i \models_T \varphi_2 \\
&\pathassign,i \models_T \varphi_1 \land \varphi_2         \quad && \text{iff} \quad&&\pathassign,i \models_T \varphi_1 \text{ and } \pathassign,i \models_T \varphi_2 \\
&\pathassign,i \models_T \X \varphi                && \text{iff} &&\pathassign,i+1 \models_T \varphi \\
&\pathassign,i \models_T \varphi_1 \U \varphi_2             && \text{iff} &&\exists k \geq i \text{ s.t. }  \pathassign, k \models_T \varphi_2 
\text{ and }
\forall m: i \leq m < k  \Rightarrow \pathassign,m \models_T \varphi_1 \\
&\pathassign,i \models_T \varphi_1 \W \varphi_2             && \text{iff} &&\forall k \geq i :  \pathassign, k \models_T \varphi_1 
\text{ or }
\exists m: i \leq  m \leq k  : \pathassign,m \models_T \varphi_2 \\
&\pathassign,i \models_T \exists \pi \ldot \varphi && \text{iff} && \pathassign[\pi \mapsto t], i \models_T \varphi \text{ for some $t \in T$}\\
&\pathassign,i \models_T \forall \pi \ldot \varphi && \text{iff} && \pathassign[\pi \mapsto t], i \models_T \varphi \text{ for all $t \in T$}\\
&\pathassign,i \models_T \existsu p \ldot \varphi      \quad &&\text{iff} \quad &&  \pathassign[p \mapsto s], i \models_{T[p \mapsto s]} \varphi 
\text{ for some $s \in (\pow{\set{p}})^\omega$} \\
&\pathassign,i \models_T \forallu p \ldot \varphi      \quad &&\text{iff} \quad && \pathassign[p \mapsto s], i \models_{T[p \mapsto s]} \varphi
\text{ for all $s \in (\pow{\set{p}})^\omega$}\\
&\pathassign,i \models_T \existsso p \ldot \varphi      \quad &&\text{iff} \quad && \pathassign', i \models_{T'} \varphi \text{ for some $T' \subseteq (\pow{\ap})^\omega$}
\text{and $\pathassign'\colon \pathvars \to T'$ such that}\\
&&&&& T =_{\ap\setminus\{p\}} T' \text{ and } \pathassign =_{\ap\setminus\{p\}} \pathassign' \\
&\pathassign,i \models_T \forallso p \ldot \varphi      \quad &&\text{iff} \quad && \pathassign', i \models_{T'} \varphi \text{ for all $T' \subseteq (\pow{\ap})^\omega$}
\text{and $\pathassign'\colon \pathvars \to T'$ such that} \\
&&&&& T =_{\ap\setminus\{p\}} T' \text{ and } \pathassign =_{\ap\setminus\{p\}} \pathassign' 
\end{alignat*}
In the sequel, we describe fragments of \HQPTLP by restricting the quantifier prefixes of formulae. We use $\exists_\pi \,/\, \forall_\pi$ to denote trace quantification, $\existsu_p \,/\, \forallu_p$ for uniform propositional quantification, and $\existsso_p \,/\, \forallso_p$  for non-uniform propositional quantification. We use $\exists$ ($\forall$, resp.) if we do not need to distinguish between the different types of existential (universal, resp.) quantifiers. We write $Q$ to refer to both $\exists$ and $\forall$. For a logic $L$ and a regular expression $e$, we write $eL$ to denote the set of $L$-formulae whose quantifier prefixes are generated by $e$. E.g., $\forall^*\exists^*\HQPTL$ refers to $\HQPTL$-formulae with quantifier prefix $\{\forallu_p,\forall_\pi\}^*\{\existsu_p,\exists_\pi\}^*$.

Next we relate the expressivity of extensions of $\teamltl$ to fragments of $\HQPTLP$.
We show that $\teamltl(\clor,\flatop)$ and $\teamltl(\clor,{\cneg\!\bot}, \flatop)$ can be translated to the prefix fragments  $\existsso_p \quantu^*_p \exists_\pi^*\forall_\pi$ and $\existsso_p \quantu^*_p \forall_\pi$ of $\HQPTLP$.
The translations provide insight into the limits of the expressivity of different extensions of $\teamltl$.
In particular, they show that in order to simulate the generation of subteams with the $\lor$-operator in \teamltl, one existential second-order quantifier $\exists_p$ is sufficient.
Meanwhile, the difference between downward closed team properties and general team properties manifests itself by a different need for trace quantifiers: for downward closed properties, a single $\forall_\pi$ quantifier is enough, whereas in the general case, a $\exists_\pi^*\forall_\pi$ quantifier alternation is needed.

As a prerequisite for the translation, we establish that evaluating $\teamltl(\clor, {\cneg\!\bot}, \flatop)$-formulae can only create countably many different teams.
For a given team $(T,i)$ and $\teamltl(\clor, {\cneg\!\bot}, \flatop)$-formula $\varphi$, the verification of $(T,i)\models \varphi$ boils down to checking statements of the form $(S,j)\models \psi$, where $S \in \mathcal{S}_T\subseteq \pow{T}$ for some set $\mathcal{S}_T$, $j\in \N$, and $\psi$ is an atomic formula, together with expressions of the form $S_1 =  S_2 \cup S_3$, where $S_1, S_2, S_3\in \mathcal{S}_T$. The following lemma, proven in~\ifbool{appdx}{Appendix \ref{a:lem:countableteams}}{the full version of this paper \cite{arxivVersion}}, implies that the set $\mathcal{S}_T$ can be fixed as a countable set that depends only on~$T$.

\begin{restatable}{lemma}{lemcountableteams}\label{lem:countableteams}
	For every set  $T$ of traces over a countable $\ap$, there exists a countable $\mathcal{S}_T\subseteq \pow{T}$ such that, for every $\teamltl(\clor, {\cneg\!\bot}, \flatop)$-formula $\varphi$ and $i\in\N$,
		$(T,i) \models \varphi \quad\text{iff}\quad (T,i) \models^* \varphi$,
	where the satisfaction relation $\models^*$ is defined the same way as $\models$ except that in the semantic clause for $\lor$ we require additionally that the two subteams $T_1,T_2\in \mathcal{S}_T$.
\end{restatable}
Using of the above lemma, we obtain translations from the most interesting extensions of $\teamltl$ to weak prefix fragments of $\HQPTLP$; for details and proofs see~\ifbool{appdx}{Appendix \ref{a:lem:countableteams}}{the full version of this paper \cite{arxivVersion}}.
%
%

\begin{restatable}{theorem}{teamltltohyperqptlthm}
	\label{s:teamltl2hyperqptl}
	For every $\varphi\in \teamltl(\clor, {\cneg\!\bot}, \flatop)$ there exists an equivalent  
	{\rm$\HQPTLP$}-formula $\varphi^*$ in the $\existsso_p \quantu^*_p \exists_\pi^*\forall_\pi$ fragment. If $\varphi\in \teamltl(\clor,\flatop)$, $\varphi^*$ can be defined in the $\existsso_p \quantu^*_p \forall_\pi$ fragment.
	The size of $\varphi^*$ is linear w.r.t. the size of $\varphi$.
\end{restatable}

\section{Decidable fragments of \teamltl}
In this section, we further study the expressivity landscape between the frameworks of \teamltl and \hyltl. We utilise these connections to prove decidability of the model checking problem of certain variants of \teamltl. 
We compare 
the expressivity of extensions of $\teamltl$ that satisfy certain semantic invariances to that of $\uhyltl$ and $\existsu_p^*\forall_\pi\HQPTL$.
Thereby, we provide a partial answer to an open 
problem posed in \cite{kmvz18} concerning the complexity of the model checking problem of $\teamltl$ and its extensions. The problem 
is known to be in $\PSPACE$ for the fragment of $\teamltl$ without $\lor$  \cite{kmvz18}. However, for $\teamltl$ with $\lor$, no meaningful upper bounds for the problem was known before. The best previous upper bound could be obtained from $\teamltl(\cneg)$, for which the problem is highly undecidable \cite{LUCK2020}. 
The reason for this lack of results is that developing algorithms for team logics with $\lor$ turned out to be comparatively hard.
The main source of difficulty is that the semantic definition of $\lor$ does not yield any reasonable compositional brute force algorithm: the verification of $(T,i)\models \varphi\lor\psi$ with $T$ generated by a finite Kripke structure proceeds by checking that $(T_1,i)\models \varphi$ and $(T_2,i)\models \psi$ for some $T_1\cup T_2 = T$,
but it can well be that $T_1$ and $T_2$ cannot be generated from any finite Kripke structure whatsoever.
%
The main results of this section are the decidability of the model checking problems of the \emph{$k$-coherent fragment} of $\teamltl(\cneg)$ and the \emph{left-flat fragment} of $\teamltl(\clor, \flatop)$. 
We obtain inclusions to $\EXPSPACE$ by translations to $\uhyltl$ and $\existsu_p^*\forall_\pi\HQPTL$.

\subsection{The k-coherent fragment and $\uhyltl$}
The universal fragment of HyperLTL is one of the most studied fragments as it contains the set of \emph{safety} hyperproperties expressible in HyperLTL~\cite{Clarkson+Schneider/10/Hyperproperties}. In particular, formulae of the form $\forall \pi_1 \ldots \forall \pi_k\ldot \psi$ state $k$-safety properties (if $\psi$ is a safety LTL formula)~\cite{DBLP:conf/csfw/FinkbeinerHT19}, where non-satisfying trace sets contain bad prefixes of at most $k$ traces. 
In general, $\forall^k\hyltl$ formulae satisfy the following inherent invariance:
\(
\emptyset,i \models_T \varphi  \quad \text{iff} \quad \emptyset,i \models_{T'} \varphi, \text{ for all $T'\subseteq T$ s.t. $\lvert T'\rvert \leq k$}.
\)
That is, a $\forall^k\hyltl$-formula $\phi$ is satisfied by a trace set $T$, iff it is satisfied by all subsets of $T$ of size at most $k$.
This property is called $k$-\emph{coherence} in the team semantics literature \cite{Kontinenj13}. The main result of this section is that all $k$-coherent properties expressible in 
$\teamltl(\cneg)$ are expressible in $\forall^k\hyltl$. This implies that, with respect to trace properties, 
all logics between $\teamltl(\flatop)$ and 
$\teamltl(\cneg)$ are equi-expressive to $\forall\hyltl$, i.e., LTL.

\begin{definition}
Let $\mathcal{A}$ be any collection of atoms and connectives introduced so far. A  formula $\varphi$ in $\teamltl(\mathcal{A})$ is said to be \emph{$k$-coherent} ($k\in\N$) if for every team $(T,i)$,
\[
(T,i) \models \varphi \quad \text{iff} \quad (S,i)\models \varphi \text{ for every $S\subseteq T$ with $\lvert S \rvert \leq k$}.
\]
\end{definition}
%
%


We will next show that, with respect to $k$-coherent properties, $\teamltl(\cneg)$ is at most as expressive as  $\forall^k\hyltl$. We define a translation from $\teamltl(\cneg)$ to $\uhyltl$ that preserves the satisfaction relation with respect to teams of bounded size. 
Given a finite set $\Phi$ of trace variables, the translation is defined as follows:
\begin{align*}
p^\Phi  &\dfn \bigwedge_{\pi\in\Phi} p_\pi &&& 
(\neg p)^\Phi  &\dfn \bigwedge_{\pi\in\Phi} \neg p_\pi &&&
(\cneg \phi)^\Phi  &\dfn  \neg\phi^{\Phi} \\
(\X \phi)^\Phi &\dfn \X \phi^\Phi &&&
(\phi \land \psi)^{\Phi} &\dfn  \phi^{\Phi} \land \psi^{\Phi} &&&
(\phi \lor \psi)^\Phi &\dfn  \bigvee_{\Phi_0 \cup \Phi_1 = \Phi}\phi^{\Phi_0} \land \psi^{\Phi_1} \\
(\phi \U \psi)^\Phi &\dfn \phi^\Phi \U \psi^\Phi &&&
(\phi \W \psi)^\Phi &\dfn \phi^\Phi \W \psi^\Phi &&& &
\end{align*}
where $\neg \phi^{\Phi}$ stands for the negation of $\phi^{\Phi}$ in negation normal form.
The following lemma, from which the subsequent theorem follows, is proved by induction. See~\ifbool{appdx}{Appendix \ref{a:ltl2hyper}}{the full version of this paper \cite{arxivVersion}} for detailed proofs.
\begin{restatable}{lemma}{ltlttohyperlm}
\label{lemma:ltl2hyper}
Let  $\phi$ be a formula of $\teamltl(\cneg)$ and $\Phi=\{\pi_1,\dots, \pi_k\}$ a finite set of trace variables. For any team $(T,i)$  with $\vert T\rvert \leq k$, any set $S\supseteq T$  of traces, and any assignment $\Pi$ with $\Pi[\Phi] = T$, we have that
	$(T,i) \models \phi \text{ iff } \Pi,i \models_S \phi^{\Phi}$.
Furthermore, if $\phi$ is  downward closed and $T\neq \emptyset$, then
\(
(T,i) \models \phi \text{ iff } \emptyset,i \models_T \forall \pi_1\dots \forall\pi_k \ldot \phi^{\Phi}.
\)		
\end{restatable}
\begin{restatable}{theorem}{kcoherentteamhyper}
\label{thm:kcoherent-team-hyper}
Every $k$-coherent property that is definable in $\teamltl(\cneg)$ is also definable in $\forall^k\hyltl$. 
\end{restatable}
Since model checking for $\uhyltl$ is $\PSPACE$-complete and its data complexity 
(model checking with a fixed formula)
is $\mathrm{NL}$-complete \cite{conf/cav/FinkbeinerRS15}, and since the above translation from ${\teamltl(\cneg)}$ to $\uhyltl$ is exponential for any $k$, we get the following corollary:
\begin{corollary}\label{cor:kcoherentMC}
For any fixed $k\in\N$, the model checking problem for $\teamltl(\cneg)$, restricted to $k$-coherent properties, is in $\EXPSPACE$, and in $\mathrm{NL}$ for data complexity.
\end{corollary}
Clearly $(T,i) \models \flatop \varphi$ \text{iff} $\emptyset,i \models_T \forall \pi \ldot \varphi_\pi$ for any $\varphi \in \LTL$, and hence we obtain the following:
\begin{corollary}
\label{1coh_HyperLTL}
The restriction of $\teamltl(\flatop)$ to formulae of the form $\flatop \varphi$ is expressively equivalent to $\forall\hyltl$.
\end{corollary}

While model checking for $k$-coherent properties is decidable, 
checking whether a given formula defines a $k$-coherent property is not, in general, decidable.
\begin{restatable}{theorem}{thmundeccoh}\label{thm:undeccoh}
Checking whether a $\teamltl(\subseteq, \clor)$-formula is $1$-coherent is undecidable.
\end{restatable}
\begin{proof}
The idea of the undecidability proof is as follows: Given any $\teamltl(\subseteq, \clor)$-formula $\varphi$, we can use a simple rewriting rule to obtain an $\LTL$-formula $\varphi^*$ such that $\varphi$ is not satisfiable (in the sense of $\teamltl$) if and only if $\varphi$ is $1$-coherent and $\varphi^*$ is not satisfiable (in the $\LTL$-sense).
Now, since checking $\LTL$-satisfiability can be done in $\PSPACE$~\cite{SistlaC85} and non-satisfiability for $\teamltl(\subseteq, \clor)$ is $\Pi^0_1$-hard by Corollary \ref{cor:undecsat}, it follows that checking 1-coherence is $\Pi^0_1$-hard as well.
For a detailed proof, see~\ifbool{appdx}{Appendix \ref{a:thm:undeccoh}}{the full version of this paper \cite{arxivVersion}}.
\end{proof}

The same holds for any extension of $\teamltl$ with an undecidable satisfiability or validity problem and whose formulae can be computably translated to $\teamltl$ while preserving satisfaction over singleton teams.

\subsection{The left-flat fragment and $\HQPTLP$}\label{sec:decidable}

In this subsection, we show that formulae $\phi$ from the left-flat fragment  of $\teamltl(\clor, \flatop)$ (defined below) can be translated to \HQPTL formulae that are linear in the size of $\phi$. 
The known model checking algorithm of \HQPTL~\cite{MarkusThesis} then immediately yields a model checking algorithm for the left-flat fragment of $\teamltl(\clor, \flatop)$. 

\begin{definition}[The left-flat fragment]\label{left_flat_df}
Let $\mathcal{A}$ be a collection of atoms and connectives.
A $\teamltl(\mathcal{A})$-formula belongs to the left-flat fragment if in each of its subformulae of the form $\psi \U \phi$ or $\psi \W \phi$, $\psi$ is a flat formula (as defined in \Cref{sec:teamltl}).
\end{definition}
Such defined fragment 
allows for arbitrary use of the $\LTLeventually$ operator, and therefore remains incomparable to \hyltl~\cite{kmvz18}. 
For instance, 
$
\LTLeventually \dep(a,b) \lor \LTLeventually \dep(c,d)
$ is a nontrivial formula in this fragment.
It states that the set of traces can be partitioned into two parts, one where eventually $a$ determines the value of $b$, and another where eventually $c$ determines the value of $d$.
The property is not expressible in \hyltl, because 
\hyltl cannot state the property ``there is a point in time at which $p$ holds on all (or infinitely many) traces''~\cite{conf/fossacs/BozzelliMP15}.

It follows from Theorem \ref{thm:undeccoh} that checking whether a $\teamltl(\subseteq, \clor)$-formula belongs to the left-flat fragment is undecidable (as flatness equals 1-coherency). Nevertheless, a decidable syntax for left-flat formulae can be obtained by using the operator $\flatop$. Formulae $\flatop \psi$ are always flat and equivalent to $\psi$ if $\psi$ is flat. Therefore, in Definition \ref{left_flat_df}, instead of imposing the semantic condition of $\psi$ being flat in subformulae $\psi \U \phi$ and $\psi \W \phi$, we could require that the subformulae must be of the form $(\flatop \psi) \U \phi$ or $(\flatop \psi) \W \phi$.

We now describe a translation from the left-flat fragment of $\teamltl(\clor,\flatop)$ to the $\existsu_p^*\forall_\pi$ fragment of \HQPTL. 
In this translation, we make use of the fact that satisfaction of flat formulae $\phi$ can be determined with the usual (single-traced) \LTL semantics. 
In the evaluation of $\phi$, it is thus sufficient to consider only finitely many subteams, whose temporal behaviour can be reflected by existentially quantified~$q$-sequences. The quantified sequences refer to points in time, at which subformulae have to hold for a trace to belong to a team.

A left-flat $\teamltl(\clor,\flatop)$-formula $\varphi$ will be translated into a formula with existential propositional quantifiers followed by a single trace quantifier.
The existential propositional quantifiers either indicate a point in time at which a subformula of $\phi_i$ is evaluated or resolve the decision for $\clor$-choices. For subformulae, we use propositions $r^{\phi_i}$, and if the same subformula occurs multiple times, it is associated with different $r^{\phi_i}$. For the resolution of $\clor$-choices we use propositions $d^{\phi \clor \psi}$.
Additionally, $r$ is a free proposition for the point in time at which $\varphi$ is to be evaluated.
The universal quantifier $\forall \pi$ sorts each trace into one of the finitely many teams.

Let $\forall \pi \ldot \hat{\psi}$ be the \hyltl formula given by \Cref{thm:kcoherent-team-hyper} for any flat formula $\psi$ (since a flat formula is $1$-coherent). 
We translate $\varphi$ inductively with respect to $r$:
\begin{alignat*}{3}
&[p, r] &&~\coloneqq~&& \G (r_\pi \rightarrow p_\pi) \\
&[\neg p, r] &&~\coloneqq~&& \G (r_\pi \rightarrow \neg p_\pi) \\
&[\X \varphi, r] &&~\coloneqq~&& \G(r_\pi \leftrightarrow \X r^{\phi}_{\pi}) \land [\varphi, r^{\phi}] \\
&[\flatop \varphi, r] &&~\coloneqq~&& \G (r_\pi \rightarrow \hat\phi) \\
&[\varphi \land \psi, r] &&~\coloneqq~&& [\varphi, r] \land [\psi, r] \\
&[\varphi \lor \psi, r] &&~\coloneqq~&& [\varphi, r] \lor [\psi, r] \\
&[\varphi \clor \psi, r] &&~\coloneqq~&& (d^{\phi \clor \psi}_\pi \rightarrow [\varphi, r]) \land (\neg d^{\phi \clor \psi}_\pi \rightarrow [\psi, r])\\
&[\varphi \W \psi, r] &&~\coloneqq~&& \G (r_\pi \rightarrow r^{\phi}_\pi \W (r^{\psi}_\pi \land \X \G \neg r^{\psi}_\pi))
\end{alignat*}
Now, let $r^1, \ldots r^n$ be the free propositions occurring in $[\varphi, r]$ and $\pi$ the free trace variable. Define the following $\existsu_p^*\forall_\pi$ \HQPTL formula:
$
\existsu r \existsu r^1 \ldots \existsu r^n \ldot \forall \pi \ldot r_\pi \land \X\G \neg r_\pi \land [\phi, r].
$
Correctness of the translation can be argued intuitively as follows.
The left-flat formula $\phi$ can be evaluated independently from the other traces in a team.
Therefore, the operators $\U$ and $\W$, whose right-hand sides argue only about a single point in time, can only generate finitely many teams. Thus, there are only finitely many points of synchronization, all of which are quantified existentially.
Every trace fits into one of the teams described by the quantified propositional variables. We verify that the translation is indeed correct in~\ifbool{appdx}{Appendix \ref{a:correctness}}{the full version of this paper \cite{arxivVersion}}.
As the construction in \Cref{thm:kcoherent-team-hyper} yields a formula $\hat{\phi}$ whose size is linear in the original formula $\phi$, the translation is obviously linear. We therefore state the following theorem.

\begin{restatable}{theorem}{teamltltoeahyperqptl}
\label{thm:translation}
For every formula $\phi$ from the left-flat fragment of $\teamltl(\clor, \flatop)$, we can compute an equivalent $\existsu_p^*\forall_\pi$ \HQPTL formula of size linear in the size of $\phi$.
\end{restatable}

Recall that the model checking problem of \hyltl formulae with one quantifier alternation is $\EXPSPACE$-complete~\cite{conf/cav/FinkbeinerRS15} in the size of the formula, 
and $\PSPACE$-complete in the size of the Kripke structure 
\cite{conf/cav/FinkbeinerRS15}.
These results directly transfer to \HQPTL~\cite{MarkusThesis} (in which \HQPTL was called \textit{\hyltl with extended quantification} instead): for model checking a \HQPTL formula, the Kripke structure can be extended by two states  generating all possible $q$-sequences.
Since the translation from $\teamltl(\clor, \flatop)$ to \HQPTL yields a formula in the $\existsu_p^*\forall_\pi$ fragment with a single quantifier alternation and preserves the size of the formula, we obtain the following theorem. 

\begin{theorem}\label{thm:leftflatMC}
The model checking problem for left-flat $\teamltl(\clor, \flatop)$-formulae is in $\EXPSPACE$, and in $\PSPACE$ for data complexity.
\end{theorem}

\section{Conclusion}
We studied \teamltl under the synchronous semantics. \teamltl is a powerful but not yet well-studied logic that can express hyperproperties without explicit quantification over traces or propositions. As such, properties which need various different quantifiers in traditional (quantification-based) hyperlogics become expressible in a concise fashion.
One of the main advantages of \teamltl is the ability to equip it with a range of atomic statements and connectives to obtain logics of varying expressivity and complexity.

We systematically studied \teamltl with respect to two of the main questions related to logics: the decision boundary of its model checking problem and its expressivity compared to other logics for hyperproperties.
We related the expressivity of \teamltl to the hyperlogics \hyltl, \HQPTL, and \HQPTLP, which are obtained by extending the traditional temporal logics \LTL and \QPTL with trace quantifiers.
We discovered that the logics $\teamltl(\clor, \flatop)$ and $\teamltl(\clor, {\cneg\!\bot}, \flatop)$ are expressively complete with respect to all downward closed, and all atomic notions of dependence, respectively.
We were able to show that $\teamltl(\clor, {\cneg\!\bot}, \flatop)$ can be expressed in a fragment of \HQPTLP.
Furthermore, for $k$-coherent properties, $\teamltl(\cneg)$  is subsumed by $\forall^*$\hyltl.
Finally, the left-flat fragment of $\teamltl(\clor,\flatop)$ can be translated to \HQPTL. The last two results induce efficient model checking algorithms for the respective logics.
In addition, we showed that model checking of $\teamltl(\subseteq, \clor)$ is already undecidable, and that the additional use of the $\asub$ quantifier makes the problem highly undecidable.

We conclude with some open problems and directions for future work: 
What is the complexity of model checking for \teamltl (with the disjunction $\lor$ but without additional atoms and connectives)? Is it decidable, and is there a translation to \HQPTL?
An interesting avenue for future work is also to explore team semantics of more expressive logics than $\LTL$ such as linear time $\mu$-calculus, or branching time logics such as $\CTL^*$ and the full modal $\mu$-calculus.

\bibliography{biblio}

\ifbool{appdx}{\appendix

\section{Generalised Atoms in TeamLTL}\label{a:sec:GA}

Section \ref{sec:undec} established that already very weak fragments of $\teamltl(\cneg)$ have highly undecidable model checking. One way to obtain expressive, but computationally well behaved, team-logics for hyperproperties is to carefully introduce novel atoms stating those atomic properties that are of the most interest, and then seeing whether decidable logics, or fragments, with those atoms can be uncovered.
In Section \ref{sec:teamltl} we already saw a few examples of important atomic statements; namely dependence and inclusion atoms.
More generally, arbitrary properties of teams induce  {\em generalised atoms}. These  atoms were first introduced, in the first-order team semantics setting, by Kuusisto \cite{Kuusisto15} using generalised quantifiers. 
Propositional generalised atoms essentially encode second-order truth tables.

\begin{definition}[Generalised atoms for $\LTL$]
	An \emph{$n$-ary generalised atom} 
	is an  $n$-ary operator $\#_G(\phi_1,\dots, \phi_n)$, with an associated nonempty set $G$ of $n$-ary relations over the Boolean domain $\{0,1\}$, that applies only to $\LTL$-formulae $\phi_1,\dots, \phi_n$. Its team semantics is defined as:
	\[
		(T,i) \models \#_G(\phi_1,\dots, \phi_n)  \quad\text{\normalfont iff }~~ \{(\truth{\phi_1}_{(t,i)}, \dots,  \truth{\phi_n}_{(t,i)} ) \mid t\in T  \} \in G.
	\]
\end{definition}
E.g., dependence atoms of type $\dep(\varphi,\psi)$ can be expressed as generalised atoms via the second-order truth table $G\dfn\{  A \subseteq  \{0,1\}^2 \mid \text{ if $(a, b_1),(a,b_2) \in A$ then $b_1=b_2$} \}$. Thus, when considered as a generalised atom, the dependence atom is, in fact, a collection of generalised atoms; one for each arity.
Note that $\flatop$ can be interpret as a downward closed generalised atom where $G\dfn\{\emptyset, \{1\}\}$, when applied to $\teamltl$-formulae. From now on, we treat $\flatop$ as a generalised atom.
Often this distinction has no effect, for example, $\teamltl(\clor,\flatop)$-formula $\flatop \varphi$ is equivalent with the $\teamltl(\flatop)$-formula $\flatop \varphi'$, where $\varphi'$ is obtained from $\varphi$ by replacing all $\clor$ by $\lor$, and simply eliminating the symbols $\flatop$ from the formula.  
We denote by $\allatoms$ ($\dcatoms$, resp.) the collection of all (all downward closed, resp.) generalised atoms. 

%

It is straightforward to verify (by induction) that for any collection $\mathcal A$ of downward closed atoms and connectives (a connective is downward closed, if it preserves downward closure), the logic $\teamltl(\mathcal A)$ is downward closed as well. For instance,  $\teamltl(\dep,\clor,\flatop)$ is downward closed.

Propositional and modal logic with $\clor$ are known to be expressively complete with respect to nonempty downward closed properties (that are closed under so-called \emph{team bisimulations})
\cite{HellaLSV14, YANG2016}.
The next proposition establishes that, in the \LTL setting, $\teamltl(\clor,\flatop)$ is very expressive among the downward closed logics; all downward closed atoms 
can be expressed in the logic. 
The translation given in the proposition is inspired by an analogous one given in  \cite{YANG2016} for propositional team logics. 

\begin{restatable}{proposition}{gdtoborprop}
		\label{gd2bor}
		For any $n$-ary generalised atom $\#_G$ and $\LTL$-formulae $\varphi_1,\dots, \varphi_n$,
		we have that 
		\[
			\#_G(\varphi_1,\dots,\varphi_n)\equiv \Clor_{R\in G}\bigvee_{(b_1,\dots,b_n)\in R}\big(\flatop(\varphi_1^{b_1}\wedge\dots\wedge \varphi_n^{b_n})\wedge\cneg\!\!\bot\big)
		\]
		where $\varphi_i^{1}:=\varphi_i$ and $\varphi_i^0:=\neg \varphi_i$ in negation normal form. 
		If $\#_G$ is downward closed, the above translation can be simplified to
		\[
			\#_G(\varphi_1,\dots,\varphi_n)\,\equiv\, \Clor_{R\in G}\bigvee_{(b_1,\dots,b_n)\in R}\flatop(\varphi_1^{b_1}\wedge\dots\wedge \varphi_n^{b_n}).
		\]
		Consequently, $\teamltl(\dcatoms, \clor) \equiv  \teamltl(\clor,\flatop)$ and
		$\teamltl(\allatoms, \clor) \!\equiv  \!\teamltl(\clor,\flatop, {\cneg\!\!\bot})\!\leq\!\teamltl(\cneg)$,
		where $\cneg\!\!\bot$ is read as a generalised atom stating that the team is non-empty. 
		The elimination of each atom yields a doubly exponential disjunction over linear sized formulae.
		%
\end{restatable}
\begin{proof}
	For any $\vec{b}=(b_1,\dots,b_n)\in R\in G$, put 
	\[\vec{\varphi}^{\,\vec{b}}=\varphi_1^{b_1}\wedge\dots\wedge \varphi_n^{b_n}.\] 
	Put also 
	\[\truth{\vec{\phi}\,}_{(t,i)}=(\truth{\phi_1}_{(t,i)}, \dots,  \truth{\phi_n}_{(t,i)} ).\]
	for  any trace $t$. Now, by the empty team property, we have 
	\begin{align*}
		(S,i)\models \flatop\vec{\varphi}^{\,\vec{b}}&\Leftrightarrow\, S=\emptyset\text{ or }  
		\forall t\in S: \truth{\vec{\phi}\,}_{(t,i)}=\vec{b}\\
		&\Leftrightarrow\, \{\truth{\vec{\phi}\,}_{(t,i)}\mid t\in S\}\subseteq\{\vec{b}\}.
	\end{align*}
	Thus, 
	\begin{align*}
		&(T,i)\models\bigvee_{\vec{b}\in R}\flatop\vec{\varphi}^{\,\vec{b}} \\
		\Leftrightarrow~& \forall \vec{b}\in R, \exists T_{\vec{b}}\text{ s.t. }
		T=\bigcup_{\vec{b}\in R}T_{\vec{b}} \text{ and }\{\truth{\vec{\phi}\,}_{(t,i)}\mid t\in T_{\vec{b}}\}\subseteq\{\vec{b}\}\\
		\Leftrightarrow~& \{\truth{\vec{\phi}\,}_{(t,i)}\mid t\in T\}\subseteq R.
	\end{align*}
	Since $(S,i)\models\cneg\!\!\bot\Leftrightarrow S\neq\emptyset$, we have similarly that 
	\[(S,i)\models \flatop\vec{\varphi}^{\,\vec{b}}\wedge\cneg\!\!\bot\,\Leftrightarrow \,\{\truth{\vec{\phi}\,}_{(t,i)}\mid t\in S\}=\{\vec{b}\},\]
	and thus
	\[
	(T,i)\models\bigvee_{\vec{b}\in R} \big( \flatop\vec{\varphi}^{\,\vec{b}} \wedge\cneg\!\!\bot \big) \Leftrightarrow \{\truth{\vec{\phi}\,}_{(t,i)}\mid t\in T\}=R.
	\]
	Finally, if $\#_G$ is downward closed, then
	\begin{align*}
		&(T,i)\models\Clor_{R\in G}\bigvee_{\vec{b}\in R}\flatop\vec{\varphi}^{\,\vec{b}}\\
		\Leftrightarrow~&(T,i)\models\bigvee_{\vec{b}\in R}\flatop\vec{\varphi}^{\,\vec{b}}\text{ for some }R\in G\\
		\Leftrightarrow~&\{\truth{\vec{\phi}\,}_{(t,i)}\mid t\in T\}\subseteq R\text{ for some }R\in G\\
		\Leftrightarrow~& \{\truth{\vec{\phi}\,}_{(t,i)} \mid t\in T  \} \in G\tag{$\because$ $\#_G$ is downward closed}\\
		\Leftrightarrow~& (T,i)\models \#_G(\varphi_1,\dots,\varphi_n).
	\end{align*}
	Similarly, if no additional condition is assumed for $\#_G$, then $\displaystyle(T,i)\models\Clor_{R\in G}\bigvee_{\vec{b}\in R}(\flatop\vec{\varphi}^{\,\vec{b}}\wedge\sim\!\!\bot)\Leftrightarrow\, (T,i)\models \#_G(\varphi_1,\dots,\varphi_n)$.
\end{proof}

The operator $\flatop$ used in the above translation is essential for the result. We now illustrate with an example that the logic $\teamltl(\clor)$ without the operator $\flatop$ is strictly less expressive than $\teamltl(\clor, \flatop)$ or $\teamltl(\dcatoms,\clor)$.
\begin{proposition}\label{prop:flatopne}
	The formula $\flatop \LTLeventually p$ is not expressible in $\teamltl(\clor)$. Thus $\teamltl(\clor) < \teamltl(\clor, \flatop)$.
\end{proposition}
\begin{proof}
	Define trace sets 
	\[T_{\mathbb{N^+}} \dfn \{\, \{\}^n\{p\}\{\}^\omega \mid n\in\N\,\},~~T_{0} \dfn \{\, \{\}^\omega \,\},\]
	and $T_{\mathbb{N}} \dfn T_{\mathbb{N^+}}  \cup  T_{0}$. Clearly $(T_{\mathbb{N^+}},0) \models \flatop \LTLeventually p$ and $(T_{\mathbb{N}},0) \not\models \flatop \LTLeventually p$. We now show, by induction, that for every $\teamltl(\clor)$-formula $\varphi$, $i\in\N$, and infinite subset $T\subseteq T_{\mathbb{N^+}}$,
	\begin{equation*}
		(T,i) \models \varphi \,\Leftrightarrow\, (T\cup T_0,i) \models \varphi.
	\end{equation*}
	%
	%
	
	Since $\teamltl(\clor)$ is downward closed, it suffices to show the $``\Rightarrow"$ direction. 
	The base case that $\varphi=p$ or $\neg p$ follows from the observation that   $(S,i) \not\models p$ and $(S,i) \not\models \neg p$, for any infinite $S\subseteq T_{\mathbb{N}}$. We now only give details for the most interesting inductive cases.
	
	If $(T,i)\models \varphi \lor \psi$, then there exist  $T_1,T_2\subseteq T$ s.t. $T_1 \cup T_2 = T$, $(T_1,i)\models \varphi$ and $(T_2,i)\models \psi$, where either $T_1$ or $T_2$ is infinite, as $T$ is infinite. W.l.o.g. we way assume 
	that $T_1$ is infinite. Now by induction hypothesis, we have that 
	\[(T_1\cup T_0) \cup T_2 = T\cup T_0, \,\,(T_1\cup T_0,i)\models \varphi, \text{ and }(T_2,i)\models \psi.\]
	Therefore, $(T\cup T_0,i) \models \varphi \lor \psi$.

	If $(T,i)\models \varphi \U \psi$, then there exists $k\geq i$ s.t. $(T,k)\models \psi$ and $(T,j)\models \phi$ whenever $i\leq j < k$. By induction hypothesis, we have that $(T\cup T_0,k)\models \psi$ and $(T\cup T_0,j)\models \phi$ whenever $i\leq j < k$. Hence, we conclude that $(T\cup T_0,i)\models \varphi \U \psi$.
	
\end{proof}

It is a facinating open problem whether variants of the infamous Kamp's theorem can be developed for $\teamltl(\clor, \flatop)$ and $\teamltl(\cneg)$. In the modal team semantics setting variants of the famous van Benthem-Rosen characterisation theorem have been shown for the modal logic analogues $\ML(\clor)$ \cite{HellaLSV14} and $\ML(\cneg)$ \cite{KontinenMSV15} of $\teamltl(\clor, \flatop)$ and $\teamltl(\cneg)$, respectively.

\section{Missing Proofs of Section \ref{sec:undec}}\label{a:sec:undec}

\begin{claim}
	The claim (\ref{thm:undecidable_eq1}) on page \pageref{thm:undecidable_eq1} holds.
\end{claim}
\begin{proof}
	We start by defining the formula $\phi_{I,b}$ that enforces that the configurations encoded by $T[i,\infty]$, $i\in \N$, encode an accepting computation of the counter machine. 
	
	Define \(
	\phi_{I,b} \dfn  (\theta_{\mathrm{comp}} \land \theta_{b\mathrm{-rec}}) \Llor \top,
	\) where  $\Llor$ is a shorthand for the following condition: 
	\begin{multline*}
		(T,i)\models\phi\Llor\psi\text{ iff }\exists T_1,T_2 \text{ s.t. }T_1\neq\emptyset, ~T_1\cup T_2=T,
		(T_1,i)\models\phi\text{ and }(T_2,i)\models\psi.
	\end{multline*}
	The disjunction $\Llor$ can be defined using $\subseteq$, $\lor$, and a built-in trace $\{p\}^\omega$, where $p$ is a fresh proposition (see e.g., \cite[Lemma 3.4]{HellaKMV19}).
	The formula $\theta_{b\mathrm{-rec}} \dfn\,  \G\LTLeventually b$ describes the $b$-recurrence condition of the computation. The other formula $\theta_{\mathrm{comp}}$, which we define below in steps, states that the encoded computation is a legal one.

	First, define 
	\begin{align*}
		\single \dfn \G \bigwedge_{a \in\ap} (a \clor \neg a), \qquad
		\sdecrease  \dfn c_s  \lor ( \neg c_s \land \X \neg c_s), \text{ for $s\in\{l,m,r\}$}.
	\end{align*}
	The intuitive idea behind the above formulae are as follows: A team satisfying the formula $\single$ contains at most a single trace with respect to the propositions in $\ap$.
	If a team $(T,i)$ satisfies $\sdecrease$, then the number of traces in $T[i+1,\infty]$ satisfying $c_s$ is less or equal to the number of traces in $T[i,\infty]$ satisfying $c_s$. In our encoding of counters this would mean that the value of the counter $c$ in the configuration $\vec{c}_{i+1}$ is less or equal to its value in the configuration $\vec{c}_{i}$. Thus $\sdecrease$ will be handy below for encoding lossy computation.
	
	Next, for each instruction label $i$, we define a formula $\theta_i$ describing the result of the execution of the instruction:
	\begin{itemize}
		\item For the instruction  $i\colon C_l^+ \text{ goto } \{j, j'\}$, define
		{\setlength{\abovedisplayskip}{0pt}
			\setlength{\belowdisplayskip}{3pt}
			\setlength{\abovedisplayshortskip}{0pt}
			\setlength{\belowdisplayshortskip}{0pt}
			\begin{multline*}
				\theta_i \dfn \X (j\clor j') \land  \big( (\single \land \neg c_l \land \X c_l) \lor \ldecrease \big)
				\land \rdecrease \land \mdecrease.
			\end{multline*}
			\item For the instruction $i\colon C_l^- \text{ goto } \{j, j'\}$, define
			\begin{multline*}
				\theta_i \dfn \X (j\clor j') \land  \big( ( c_l \land \X \neg c_l) \Llor \ldecrease \big)
				\land \rdecrease \land \mdecrease.
			\end{multline*}
			\item For the instrunctions $i\colon C_s^+ \text{ goto } \{j, j'\}$ and $i\colon C_s^- \text{ goto } \{j, j'\}$ with $s\in\{m,r\}$, the formulae $\theta_i$ are defined analogously with the indices $l$, $m$, and $r$ permuted. 
			\item For the instruction $i\colon \text{if } C_s=0 \text{ goto }  j, \text{else goto }  j'$, define
			\begin{multline*}
				\theta_i \dfn  \big(\X(\neg c_s \land  j) \clor (\top \subseteq c_s \land \X j')\big) 
				\land  \ldecrease \land \mdecrease \land \rdecrease.
			\end{multline*}
		}
	\end{itemize}
	Finally, define \(
	\theta_{\mathrm{comp}} \dfn \G \Clor_{ i < n} ( i \land  \theta_i).
	\)
	We next describe the intuition of the above formulae.
	The left-most conjunct of $\theta_i$ for $i\colon C_l^+ \text{ goto } \{j, j'\}$ expresses that after executing the instruction $i$, the label of the next instruction is either $j$ or $j'$. The third and the fourth conjunct express that the values of counters $C_r$ and $C_m$ will not increase, but might decrease. The second conjunct expresses that the value of the counter $C_l$ might increase by one, stay the same, or decrease. The meaning of $\theta_i$ for $i\colon C_l^- \text{ goto } \{j, j'\}$ is similar. Finally, the formula  $\theta_i$ for $i\colon \text{if } C_s=0 \text{ goto }  j, \text{else goto }  j'$ expresses that, in the lossy execution of i, a) the values of the counters $C_l$, $C_m$, and $C_r$ might decrease, but cannot increase, b) the next instruction is either $j$ or $j'$, c) if the next instruction is $j$ then the value of the counter $C_s$, after the lossy execution, is $0$, and d)  if the next instruction is $j'$ then the value of the counter $C_s$, before the lossy execution, was not $0$.
	
	Next, we define the Kripke structure $\kK_I = (W,R,\eta,w_0)$ over the set of propositions $\{c_l,c_m,c_r,d,0,\dots, n-1\}$. The structure is defined such that every possible sequence of configurations of $M$ starting from $(0,0,0,0)$ can be encoded by some team $(T,0)$, where $T\subseteq \traces(\kK_I)$. Define $W \dfn \{(i,j,k,t,l) \mid 0 \leq i < n \text{ and }  j,k,t,l\in\{0,1\} \}$, $w_0 \dfn (0,0,0,0,0)$, $R \dfn W \times W$, and $\eta$ as the valuation such that $\eta\big((i,j,k,t,l)\big) \cap \{0,\dots, n-1\} = i$,
		\begin{multicols}{2}
			\begin{itemize}
				\item $c_l\in \eta\big((i,j,k,t,l)\big)$ if $j=1$,  
				\item $c_m\in \eta\big((i,j,k,t,l)\big)$ if $k=1$, 
				\item $c_r\in \eta\big((i,j,k,t,l)\big)$ if $t=1$, 
				\item and $d\in \eta\big((i,j,k,t,l)\big)$ if $l=1$.
			\end{itemize}
		\end{multicols}

	Assume first that $M$ has a $b$-recurring lossy computation and let $(\vec{c}_j)_{j\in\mathbb{N}}$ be the related sequence of configurations of $M$. Let $T\subseteq \traces(\kK_I)$ be a set of traces that encodes $(\vec{c}_j)_{j\in\mathbb{N}}$ in the way described above, and such that, for every $j\in \mathbb{N}$,
	\begin{itemize}
		\item if $\vec{c}_j=(i,v_l,v_m,v_r)$ and the instruction labelled $i$ is of the form $i\colon C_s^+ \text{ goto } \{j, j'\}$, then there is at most one $t\in T[j,\infty]$ such that $c_s\not\in t[0]$ but $c_s\in t[1]$,
		\item if $i$ is not of the form $i\colon C_s^+ \text{ goto } \{j, j'\}$, then there is no $t\in T[j,\infty]$ such that $c_s\not\in t[0]$ but $c_s\in t[1]$.
	\end{itemize}
	The aforementioned condition makes sure that the traces in $T$ that encode the incrementation of counter values do not change erratically. Clearly, such a $T$ always exists, given that $(\vec{c}_j)_{j\in\mathbb{N}}$ encodes a lossy computation. Furthermore,  since $T$ encodes $(\vec{c}_j)_{j\in\mathbb{N}}$ and the related $b$-recurrent lossy computation follows the instructions in $I$, we have that $(T,0)\models \theta_{\mathrm{comp}} \land \theta_{b\mathrm{-rec}}$. Finally, as $\emptyset \neq T\subseteq \traces(\kK_I)$, $\big(\traces(\kK_I),0\big) \models (\theta_{\mathrm{comp}} \land \theta_{b\mathrm{-rec}}) \Llor \top$ follows.
	
	Assume then that $\big(\traces(\kK_I),0\big) \models \phi_{I,b}$. Hence there exists some nonempty subset $T$ of $\traces(\kK_I)$ such that $(T,0)\models \theta_{\mathrm{comp}} \land \theta_{b\mathrm{-rec}}$. It is now easy to construct a sequence $(\vec{c}_j)_{j\in\mathbb{N}}$ of configurations that encode a $b$-recurrent lossy computation for $M$; for each $j\in\mathbb{N}$, define $\vec{c}_j=(i, v_l, v_m,v_r)$ such that $\bigcup T[j] \cap \{0,\dots, n-1\} = \{i\}$, and $\lvert \{t[j,\infty] \mid c_s\in t[j], t\in T \} \rvert = v_s$, for each $s\in\{l,m,r\}$.
\end{proof}

\sigunsat*

\begin{proof}
	The proof is analogous to that of Theorem \ref{thm:undecidable}. The following modifications are required to shift from lossy computation to non-lossy one.
	Firstly, we define a formula $\theta_\mathrm{diff} $ that expresses that if two traces differ (i.e., $t\neq t'$) then all of their postfixes differ as well (i.e., $t[j,\infty]\neq t'[j,\infty]$, for each $j \in \N$): 
	\begin{multline*}
		\theta_\mathrm{diff} \dfn \asub\Big(\G\big((c_l \clor \neg c_l) \land  (c_m \clor \neg c_m) \land (c_r \clor \neg c_r) \land (d \clor \neg d)\big)\\
		\quad\quad\quad\quad\clor \G\LTLeventually\big((\top\subseteq c_l \land \bot\subseteq c_l) \clor (\top\subseteq c_m \land \bot\subseteq c_m)\\ 
		\quad\clor (\top\subseteq c_r \land \bot\subseteq c_r) \clor (\top\subseteq d \land \bot\subseteq d) \big) \Big).
	\end{multline*}
	For $s\in\{l,m,r\}$, instead of using the formula $\sdecrease$ in Theorem \ref{thm:undecidable}, we make use of $\spreserve$ defined as:
	\[
	\spreserve  \dfn\,  \big( c_s \land \X c_s) \lor ( \neg c_s \land \X \neg c_s)\big).
	\]
	Finally, we define \[\phi_{I,b} \dfn\,   (\theta_\mathrm{diff} \land \theta'_{\mathrm{comp}} \land \theta_{b\mathrm{-rec}}) \Llor \top,\]
	where $\theta'_{\mathrm{comp}} \dfn \G \Clor_{ i < n} ( i \land  \theta'_i)$ and  $\theta'_i$ for each instruction $i$ is defined as follows:
	
	\begin{itemize}
		\item For the instruction  $i\colon C_l^+ \text{ goto } \{j, j'\}$, define
		{\setlength{\abovedisplayskip}{0pt}
			\setlength{\belowdisplayskip}{3pt}
			\setlength{\abovedisplayshortskip}{0pt}
			\setlength{\belowdisplayshortskip}{0pt}
			\begin{multline*}
				\theta'_i \dfn \X (j\clor j') \land  \big( (\single \land \neg c_l \land \X c_l) \Llor \lpreserve \big) 
				\land \rpreserve \land \mpreserve.
			\end{multline*}
			\item For the instruction $i\colon C_l^- \text{ goto } \{j, j'\}$, define
			\begin{multline*}
				\theta'_i \dfn \X (j\clor j') \land  \big( (\single \land c_l \land \X \neg c_l) \Llor \lpreserve \big)
				\land \rpreserve \land \mpreserve.
			\end{multline*}
			\item For $i\colon C_s^+ \text{ goto } \{j, j'\}$ and $i\colon C_s^- \text{ goto } \{j, j'\}$ with $s\in\{m,r\}$, the formulae $\theta'_i$ are defined analogously with the indices $l$, $m$, and $r$ permuted. 
			\item For the instruction $i\colon \text{if } C_s=0 \text{ goto }  j, \text{else goto }  j'$, define
			\begin{multline*}
				\theta'_i \dfn  \big((\neg c_s \land \X  j) \clor (\top \subseteq c_s \land \X j')\big) \land  \lpreserve \land \mpreserve \land \rpreserve.
			\end{multline*}
		}
	\end{itemize}
\end{proof}

\corundecsat*

\begin{proof}
	Given a finite Kripke structure $\kK = (W,R,\eta,w_0)$ over a finite set of propositions $\ap$, we introduce a fresh propositional variable $p_w$, for each $w\in W$, and let $\kK'$ be the unique extension of $\kK$ in which $p_w$ holds only in $w$. We then construct a formula $\theta_{\kK'}$ whose intention is to characterise the team $\big(\traces(\kK'),0\big)$. Define
	\begin{multline*}
		\theta_{\kK'} \dfn\, p_{w_0} \land \G \bigvee_{w\in W} \Bigg(p_w \land \big( \bigwedge_{w\neq v\in W} \neg p_v \big)
		\land \big( \bigwedge_{(w,v)\in R} \top\subseteq \X p_v \big) \land \big( \X \bigvee_{(w,v)\in R} p_v \big)\\
		\land \bigwedge_{p\in \eta(w)} p \land  \bigwedge_{p\in \ap\setminus \eta(w)} \neg p \bigg).
	\end{multline*}
	It is not hard to verify that the only sets of traces over $\ap\cup \{p_w \mid w\in W\}$ for which $(T,0) \models \theta_{\kK'}$ holds are the empty set  and $\traces(\kK')$. Now, given a finite Kripke structure $\kK$ and a formula $\varphi$ of $\teamltl(\subseteq, \clor)$ ($\teamltl(\subseteq, \clor, \asub)$, resp.) it holds that $(\traces(\kK),0)\models \varphi$ iff $\theta_{\kK'}  \land \varphi$ is satisfiable.
\end{proof}	

\section{Missing Proofs of Section 4}\label{a:lem:countableteams}

\setcounter{theorem}{4}
\begin{lemma}
	For every set  $T$ of traces over a countable $\ap$, there exists a countable set $\mathcal{S}_T\subseteq \pow{T}$ such that, for every $\teamltl(\clor, \cneg\!\bot, \flatop)$-formula $\varphi$ and $i\in\N$,
	\begin{equation}\label{eq:base}
		(T,i) \models \varphi \quad\text{iff}\quad (T,i) \models^* \varphi,
	\end{equation}
	where the satisfaction relation $\models^*$ is defined the same way as $\models$ except that in the semantic clause for $\lor$ we require additionally that the two subteams $T_1,T_2\in \mathcal{S}_T$.
\end{lemma}
\setcounter{theorem}{19}

\begin{proof}
	We first define inductively and nondeterministically a function 
	\(
	\Sub\colon (\pow{T}\times \N) \times \teamltl(\clor, \mathcal{A}) \to  \pow{T} 
	\)
	as follows:
	\begin{itemize}
		\item If $\varphi$ is an atomic formula or a generalised atom, define $\Sub\big((S,j),\varphi\big) \dfn \{S\}$ 
		\item 
		$\Sub\big((S,j), \X \varphi\big) \dfn \Sub\big((S,i+1), \varphi\big)$.
		\item 
		$\Sub\big((S,j), \varphi \land \psi\big) \dfn  \Sub\big((S,j), \varphi\big) \cup \Sub\big((S,j), \psi\big)$. 
		\item $\Sub\big((S,j), \varphi \clor \psi\big) \dfn  \Sub\big((S,j), \varphi\big) \cup \Sub\big((S,j), \psi\big)$. 
		
		\item $\Sub\big((S,j), \varphi \lor \psi\big)\dfn\Sub\big((S_1,j), \varphi\big) \cup \Sub\big((S_2,j), \psi\big) \cup \{S\}$, where subsets $S_1$ and $S_2$ of $S$ are guessed nondeterministically such that $S_1\cup S_2 = S$.
		\item 
		$\Sub\big((S,j), \psi \U \varphi\big) \dfn \Sub\big((S,j),\psi \W \varphi\big)
		\dfn  \bigcup_{i \geq j} \big( \Sub\big((S,i), \psi\big) \cup \Sub\big((S,i), \varphi\big) \big)$.
		
	\end{itemize}
	
	Clearly $\Sub\big((S,j),\psi \big)$ is a countable set. 
	Now, define 
	\[
	\mathcal{S}_T \dfn \bigcup_{\substack{ j\in \N\\ \psi\in \teamltl(\clor, \mathcal{A})}} \Sub\big((T,j),\psi \big).
	\]
	The set $\mathcal{S}_T$ is a countable union of countable sets, and thus itself countable. For each team $(S,j)$ and formula $\psi$, assuming  $\mathcal{S}_T\supseteq\Sub\big((S,j), \psi\big)$, it is not hard to show by induction that \eqref{eq:base} holds.
\end{proof}

\teamltltohyperqptlthm*
\begin{proof}
	Let $\qst$, $q$, and $r$ be distinct propositional variables. We define 
	a compositional translation $\tr_{(q,r)}$ 
	such that 
	for every team $(T,i)$ and $\teamltl(\clor, \cneg\!\bot, \flatop)$-formula $\varphi$,
	{
		\setlength{\abovedisplayskip}{3pt}
		\setlength{\belowdisplayskip}{-8pt}
		\begin{equation}\label{eq:translation}
			(T,i) \models \varphi \text{ iff } \Pi,i \models_T \existsso \qst \existsu q \existsu r \big(\tr_{(q,r)} (\varphi) \land \aux \big), \text{ where $\aux \dfn \forall \pi \ldot \qst_\pi \land q_\pi \land r_\pi$.}
		\end{equation}
	}
	
	We first fix some conventions.
	All quantified variables in the translation below are assumed to be fresh and distinct.
	We also assume that the uniformly quantified propositional variables $q,r,\dots$ are true  in exactly one level, that is, they satisfy the formula $\forall \pi \ldot \LTLeventually q_\pi \land \G(q_\pi \rightarrow \X\G \neg q_\pi)$, which we will omit in the presentation of the translation for simplicity. 
	%
	%
	
	The idea behind the translation is the following. Let $(T',i)$ denote the team obtained from $(T,i)$ by evaluating the quantifier $\existsso \qst$.
	\begin{itemize}
		\item 
		The variable $\qst$ is used to encode the countable set $\mathcal{S}_T$ of sets of traces given by Lemma  \ref{lem:countableteams}. To be precise, for each $i\in \N$, $\qst$  encodes the set $(\{t\in T' \mid \qst\in t[i] \} \upharpoonright \ap) \in \mathcal{S}_T$.
		\item The uniformly quantified variable $q$ in $\tr_{(q,r)}$ is used to encode an element of $\mathcal{S}_T$ using $\qst$: If $q \in t[i]$, then $q$ encodes the set $\{t\in T' \mid \qst\in t[i] \} \upharpoonright \ap$.
		\item The uniformly quantified variable $r$ in $\tr_{(q,r)}$ is used to encode the time step $i$ of a team $(T,i)$: If $r \in t[i]$, then $r$ encodes the time step $i$.
	\end{itemize}
	After fixing a suitable interpretation for $\qst$, teams $(S,i)$ can be encoded with pairs of uniformly quantified variables $(q,r)$, whenever $S\in\mathcal{S}_T$. The formula $\aux $
	expresses that the pair $(q,r)$ encodes the team $(T,i)$ in question.
	
	The translation $\tr_{(q,r)}$ is defined inductively as follows:
	{
		\setlength{\abovedisplayskip}{3pt}
		\setlength{\belowdisplayskip}{-9pt}
		\begin{align*} 
			\tr_{(q,r)}(\ell) \,\dfn\,& \forall \pi \ldot \big( \LTLeventually (q_\pi  \land \qst_\pi) \rightarrow \LTLeventually(r_\pi \land \ell_\pi) \big), \quad \text{where $\ell \in \{p, \neg p\}$} \\
			%
			\tr_{(q,r)}(\varphi\clor \psi) \,\dfn\,&  \tr_{(q,r)}(\varphi) \lor \tr_{(q,r)}(\psi), \quad 
			\tr_{(q,r)}(\varphi \land \psi) \,\dfn\, \tr_{(q,r)}(\varphi) \land \tr_{(q,r)}(\psi)\\
			\tr_{(q,r)}(\X \varphi) \,\dfn\,& \existsu r' \ldot \big( \G(r_\pi \leftrightarrow \X r'_{\pi}) \land \tr_{(q,r')}(\varphi)\big), \quad
			\tr_{(q,r)} (\cneg\!\bot) \,\dfn\, \exists \pi \ldot \LTLeventually (q_\pi  \land \qst_\pi)
		\end{align*}
	}
	{
		\setlength{\abovedisplayskip}{0pt}
		\setlength{\belowdisplayskip}{3pt}
		\begin{align*} 
			\tr_{(q,r)}(\flatop\varphi) \,\dfn\,& \forall \pi. \big( \LTLeventually (q_\pi  \land \qst_\pi) \rightarrow \LTLeventually(r_\pi \land \varphi^*_\pi) \big), \tag*{(where the $\varphi^*$ is an $\LTL$-formula}\\
			&\tag*{obtained from $\varphi$ by removing all $\flatop$ and using a careful bottom-up linear}\\
			&\tag*{procedure to eliminate all $\clor$ and $\cneg\!\bot$. Detailed exposition below the proof.)}\\
			\tr_{(q,r)}(\varphi\lor \psi) \,\dfn\,& \existsu q_1 q_2 \ldot\big(\varphi_{\cup}(q,q_1,q_2)
			\land \tr_{(q_1,r)}(\varphi) \land \tr_{(q_2,r)}(\psi)\big)\\
			\tr_{(q,r)}(\varphi\U\psi) \,\dfn\,& \existsu r' \ldot \Big( r \preceq r' \land \tr_{(q,r')}(\psi)
			\land \forallu r'' \ldot \big( ( r \preceq r'' \land r'' \prec r' )  \rightarrow \tr_{(q,r'')}(\varphi) \big)\Big)\\
			\tr_{(q,r)}(\varphi\W\psi) \,\dfn\,& \forall r' \ldot \Big(r \preceq r' \rightarrow \Big(\tr_{(q,r')}(\varphi)
			\lor \exists r'' \ldot \big(r'' \preceq r' \land \tr_{q,r''}(\psi)\big) \Big) \Big)
		\end{align*}%
	}
	where $\varphi_{\cup}(q,q',q'') \dfn \LTLeventually (q_\pi \land \qst_\pi) \leftrightarrow  \LTLeventually \big((q'_\pi \lor q''_\pi) \land \qst_\pi)\big)$, $r \prec r' \dfn \G(r_\pi \leftrightarrow \X \LTLeventually r'_{\pi})$, and  $r \preceq r' \dfn \G(r_\pi \rightarrow \LTLeventually r'_{\pi})$.
	
	%
	
	The formula $\tr_{(q,r)}(\varphi)\land \aux$ can be transformed to an equivalent prenex formula in the fragment $\quantu^*_p \exists_\pi^*\forall_\pi$, and if $\varphi\in \teamltl(\clor, \flatop)$, the translation is in the $\quantu^*_p \forall_\pi$ fragment. The translations are clearly linear. The equivalence \eqref{eq:translation} is proved by a routine inductive argument; we omit the detailed proof.
	%
	%
	%
	%
	%
	%
\end{proof}

Next we describe the procedure needed for the case $\flatop$ of the translation of Theorem \ref{s:teamltl2hyperqptl}. It suffices to establish a linear translation ${}^* \colon \teamltl(\clor, \cneg\!\bot) \rightarrow \teamltl$ such that $\flatop \varphi\equiv \flatop\varphi^*$, for every $\varphi\in \teamltl(\clor, \cneg\!\bot)$. The translation is written bottom-up, and the symbols $\clor$ and $\cneg\!\bot$ do not occur in the formulae $\psi$ and $\theta$ below. (We wish to remind the reader that all $\teamltl(\clor)$ formulae are satisfied by the empty team.) For formulae $\varphi$ in which $\cneg\!\bot$ does not occur, the translation only replaces the symbols $\clor$ by $\lor$. In other cases the translation depends on the position of $\cneg\!\bot$ in the subformulae.
Below $\psi$ and $\theta$ are allowed to be empty (i.e, $\cneg\!\bot \land \psi$ covers the case  $\cneg\!\bot$).
\begin{align*}
	(\cneg\!\bot \land \psi) \lor (\cneg\!\bot \land \theta) &\quad\mapsto\quad \cneg\!\bot \land (\psi \land \theta) \\
	(\cneg\!\bot \land \psi) \lor \theta &\quad\mapsto\quad \cneg\!\bot \land \psi \\
	(\cneg\!\bot \land \psi) \clor (\cneg\!\bot \land \theta) &\quad\mapsto\quad \cneg\!\bot \land (\psi \lor \theta) \\
	(\cneg\!\bot \land \psi) \clor \theta &\quad\mapsto\quad \psi \lor \theta \\
	(\cneg\!\bot \land \psi) \land (\cneg\!\bot \land \theta) &\quad\mapsto\quad \cneg\!\bot \land (\psi \land \theta) \\
	(\cneg\!\bot \land \psi) \land \theta &\quad\mapsto\quad \cneg\!\bot \land (\psi \land \theta) \\
	\X (\cneg\!\bot \land \psi) &\quad\mapsto\quad  \cneg\!\bot \land \X \psi \\ 
	(\cneg\!\bot \land \psi) \U (\cneg\!\bot \land \theta) &\quad\mapsto\quad  \cneg\!\bot \land (\psi \U \theta) \\
	\psi \U (\cneg\!\bot \land \theta) &\quad\mapsto\quad  \cneg\!\bot \land (\psi \U \theta) \\
	(\cneg\!\bot \land \psi) \U \theta &\quad\mapsto\quad  \psi \U \theta \\
	(\cneg\!\bot \land \psi) \W (\cneg\!\bot \land \theta) &\quad\mapsto\quad  \cneg\!\bot \land (\psi \W \theta) \\
	\psi \W (\cneg\!\bot \land \theta) &\quad\mapsto\quad  \psi \W \theta \\
	(\cneg\!\bot \land \psi) \W \theta &\quad\mapsto\quad  \psi \W \theta
\end{align*}
It is straightforward to check that, for any $\teamltl$-formulae $\psi$ and $\theta$, the above translation preserves the truth value of formulae over teams of cardinality at most $1$.
Hence, when using the translation to $\varphi$, we obtain a formula $\varphi^+$ of the form $\cneg\!\bot \land \psi$ or $\psi$, where $\psi\in\teamltl$, such that $\flatop\varphi$ and $\flatop\varphi^+$ are equivalent. Finally note that, by the semantics of $\flatop$, the formulae $\flatop(\cneg\!\bot \land \psi)$ and $\flatop\psi$ are equivalent. By defining $\varphi^* \dfn \psi$ we obtain the linear size formula needed in Theorem \ref{s:teamltl2hyperqptl}.

\section{Missing Proofs of Section 5}\label{a:ltl2hyper}

\ltlttohyperlm*

\begin{proof}
	The second claim of the lemma follows easily from the first claim. 
	We prove the first claim by induction on $\varphi$. If $\varphi=p$, then
	\begin{align*}
		(T,i)\models p &\,\Leftrightarrow\, \forall t\in T: p\in t[i]\\
		&\,\Leftrightarrow\, \forall \pi\in \Phi: p\in \Pi(\pi)[i]\tag{$\because T=\Pi[\Phi]$}\\
		&\,\Leftrightarrow\, \Pi, i\models_S\bigwedge_{\pi\in \Phi}p_\pi\tag{$\because S\supseteq T$}\\
		&\,\Leftrightarrow\, \Pi, i\models_S p^{\Phi}
	\end{align*}
	Note that if $T=\emptyset$, then $\Phi=\emptyset$ and by definition $p^\emptyset=\top$. In this special case the above proof  still goes through, and in particular $(\emptyset, i)\models p$ and $\Pi,i\models_S\top$. 
	
	The case $\varphi=\neg p$ is proved analogously. If $\varphi=\sim\psi$, then
	\begin{align*}
		(T,i)  \models \sim\psi ~ 
		&\Leftrightarrow~(T, i) \not\models \psi \\
		&\Leftrightarrow~ (\Pi, i) \not\models_S \psi^{\Phi} \tag{by induction hypothesis}\\
		&\Leftrightarrow ~ (\Pi, i) \models_S  \neg\psi^{\Phi}.
	\end{align*}
	
	If $\varphi=\psi\lor\chi$, then
	\begin{align*}
		&(T,i)  \models \psi\lor\chi  \\
		\Leftrightarrow~ &(\Pi[\Phi_0], i) \models \psi \text{ and } (\Pi[\Phi_1], i) \models \chi\text{ for some } \Phi_0,\Phi_1\\
		&\text{ with }\Phi_0 \cup \Phi_1= \Phi \tag{since $T=\Pi[\Phi]$}\\
		\Leftrightarrow~ & \Pi, i \models_S \psi^{\Phi_0} \text{ and } \Pi, i \models_S \chi^{\Phi_1}
		\text{for some } \Phi_0,\Phi_1\text{ with }\\
		&\Phi_0 \cup \Phi_1= \Phi \text{ (by IH,  $\because S\supseteq T=\Pi[\Phi]\supseteq \Pi[\Phi_0],\Pi[\Phi_1]$)}\\
		\Leftrightarrow~& \Pi, i \models_S  \bigvee_{\Phi_0 \cup \Phi_1 =\Phi} (\psi^{\Phi_0} \land \chi^{\Phi_1})\\
		\Leftrightarrow ~& \Pi, i \models_S  (\psi \lor \chi)^{\Phi}.
	\end{align*}

	If $\varphi=\chi\U\psi$, then
	\begin{align*}
		&(T, i) \models \chi \U \psi  \\
		\Leftrightarrow\, &\exists n \geq i : (T,n) \models \psi, \text{ and }
		\forall m: i\leq m < n \Rightarrow (T, m)\models \chi \\
		\Leftrightarrow\, &\exists n \geq i : \Pi, n \models_S \psi^{\Phi}, \text{ and } \forall m: i\leq m < n \Rightarrow \Pi, m \models_S \chi^{\Phi}\\
		\Leftrightarrow\, &\Pi,i \models_S \chi^{\Phi} \U  \psi^{\Phi} \\
		\Leftrightarrow\, &\Pi,i \models_S (\chi \U \psi)^{\Phi}.
	\end{align*}
	
	The other inductive cases are proved by using a routine argument.
	%
	%
	%
	%
	%
\end{proof}

\kcoherentteamhyper*
\begin{proof}
	Let $\varphi$ be a $k$-coherent $\teamltl(\cneg)$-formula.
	If \mbox{$\emptyset, i \not\models \varphi$} for some $i\in\mathbb{N}$, then $(T, j) \not\models \varphi$ for every team $(T,j)$. 
	We then translate $\varphi$ to $\forall \pi \ldot \bot$. Assume now $\emptyset, i \models \varphi$. For any nonempty team $(T,i)$, we have that
	{\setlength{\belowdisplayskip}{0pt}
		\begin{align*}
			&(T,i)\models \varphi \\
			\Leftrightarrow~ &(S,i) \models \varphi \text{ for every $S\subseteq T$ with $\vert S \rvert \leq k$}\tag{since $\varphi$ is $k$-coherent}\\
			\Leftrightarrow~& \Pi,i \models_{T} \varphi^{\{\pi_1,\dots,\pi_k\}} \text{ for every $\Pi$ s.t  $\Pi[\{\pi_1,\dots,\pi_k\}] \subseteq T$}\tag{by Lemma \ref{lemma:ltl2hyper} \& $\emptyset, i \models \varphi$}\\
			\Leftrightarrow~ &\emptyset,i \models_{T} \forall \pi_1\dots \forall\pi_k \ldot \varphi^{\{\pi_1,\dots,\pi_k\}}.
		\end{align*}
	}
\end{proof}



\label{a:thm:undeccoh}
\thmundeccoh*
\begin{proof}
	For a given $\teamltl(\subseteq, \clor)$-formula $\varphi$, let $\varphi^*$ denote the $\teamltl$-formula obtained from $\varphi$ by first replacing each disjunctions $\clor$ by $\lor$, and then replacing each inclusion atom of the form $\varphi_1,\dots,\varphi_n \subseteq \psi_1,\dots,\psi_n$ by the conjunction
	\[
	\bigwedge_{i\leq n} \varphi_i \leftrightarrow \psi_i.
	\]
	It is easy to verify, by induction, that over singleton teams  $\varphi$ and $\varphi^*$ are equivalent.
	Moreover, since \teamltl satisfies the singleton equivalence property, we obtain the following equivalences, for every trace $t$:
	\[
	(\{t\},i)\models \varphi \,\Leftrightarrow\, (\{t\},i)\models \varphi^* \,\Leftrightarrow\, (t,i)\models \varphi^*,
	\]
	where the right-most satisfaction relation is the standard one for $\LTL$.

	Next, we show that for any formula $\varphi \in \teamltl(\subseteq, \clor)$ 
	\allowdisplaybreaks[3]
	\begin{align*}
		\text{$\varphi$ is not satisfiable} \quad\Leftrightarrow\quad& \text{$\varphi$ is $1$-coherent and $\varphi^*$ is not satisfiable in the sense of $\LTL$}. 
	\end{align*}
	Assume that $\varphi$ is not satisfiable. Since no team satisfies $\varphi$, it is trivially $1$-coherent. In particular, for any trace $t$, we have $(\{t\},i)\not\models \varphi$, which then implies that $(t,i)\not\models \varphi^*$. Thus $\varphi^*$ is not satisfiable in the sense of $\LTL$. Conversely, suppose  $\varphi$ is $1$-coherent and $\varphi^*$ is not satisfiable. Then for any trace $t$, we have $(t,i)\not\models \varphi^*$, which implies that $(\{t\},i)\not\models \varphi$. It then follows, by $1$-coherence, that $(T,i)\not\models\varphi$ for every team $T$. Thus $\varphi$ is not satisfiable.
	
	Now, since checking $\LTL$-satisfiability can be done in $\PSPACE$ \cite{SistlaC85} and non-satisfiability for $\teamltl(\subseteq, \clor)$ is $\Pi^0_1$-hard by Corollary \ref{cor:undecsat}, we conclude that checking $1$-coherence is undecidable.

\end{proof}

\label{a:correctness}

\teamltltoeahyperqptl*

The theorem follows  from the following lemma.

\begin{lemma}
	Let $\varphi$ be a $\teamltl(\clor,\flatop)$-formula, and $r_1, \ldots, r_n$ free variables occurring in $[\phi,r]$ but not in $\phi$. Let $i \in \nats$, and 
	$s \in (2^{\set{r}})^\omega$  a sequence that has $r$ set exactly at position $i$. For every team $(T,i)$,
	\[(T,i) \models \varphi\,\Leftrightarrow\,\emptyset, 0 \models_{T[r \mapsto s]} \exists r_1 \ldots r_n \forall \pi \ldot [\phi,r].\] 
\end{lemma}
\begin{proof}
	We proceed by 
	induction on $\phi$.
	
		\noindent		Case for $p$:
		\begin{align*}
			(T,i) \models p &\Leftrightarrow\forall t \in T: p \in t[i] \\
			&\Leftrightarrow \emptyset, 0 \models_{T[r \mapsto s]} \forall \pi \ldot \G(r_\pi \rightarrow p_\pi)
		\end{align*}
		Case for $\neg p$: Similar to the above.
		
		\noindent			Case for $\flatop \varphi$:
		\begin{align*}
			(T,i) \models \flatop \varphi &\Leftrightarrow \forall t \in T: (t,i) \models \varphi \\
			&\Leftrightarrow \emptyset, 0 \models_T \forall \pi \ldot \G (r_\pi \rightarrow \hat\phi) \tag{by \Cref{thm:kcoherent-team-hyper}, as $\flatop \varphi$ is  1-coherent}
		\end{align*}	
		Case for $\X \phi$:
		\begin{align*}
			& (T,i) \models \X \phi \\
			\Leftrightarrow~& (T, i+1) \models \phi \\
			\overset{\text{IH}}{\Leftrightarrow}~& \emptyset, 0 \models_{T[r^{\phi} \mapsto s^{\phi}]} \exists r_1 \ldots r_n \ldot \forall \pi \ldot [\phi, r^{\phi}] \\
			\Leftrightarrow ~&\emptyset, 0 \models_{T[r \mapsto s]} \exists r^{\phi} \ldot \exists r_1 \ldots r_n \ldot \forall \pi \ldot \G(r_\pi \leftrightarrow \X r_\pi^{\phi}) \land [\phi, r^{\phi}]
		\end{align*}
		Case for $\phi \land \psi$: Easy, by induction hypothesis.
		
		\noindent			Case for $\phi \lor \psi$:
		\begin{align*}
			&(T,i) \models \phi \lor \psi \\
			\Leftrightarrow~& \exists T_1, T_2 \text{ s.t. } T = T_1 \cup T_2 \text{ and } \\
			& \qquad (T_1, i) \models \phi \text{ and } (T_2, i) \models  \psi \\
			\overset{\text{IH}}{\Leftrightarrow}~& \emptyset, 0 \models_{T_1[r \mapsto s]} \exists r_1^1 \ldots r_n^1 \ldot \forall \pi \ldot [\phi, r] \text{ and } \emptyset, 0 \models_{T_2[r \mapsto s]} \exists r_1^2 \ldots r_m^2 \ldot \forall \pi \ldot [\psi, r]\\
			\Leftrightarrow ~&\emptyset, 0 \models_{T[r \mapsto s]} \exists r_1^1 \ldots r_n^1 r_1^2 \ldots r_m^2 \ldot \forall \pi \ldot [\phi, r] \lor [\psi,r] \tag{since $T_1[r \mapsto s] \cup T_2[r \mapsto s] = T[r \mapsto s]$}
		\end{align*}
		Case for $\phi \clor \psi$:
		\begin{align*}
			&(T,i) \models \phi \clor \psi \\
			\Leftrightarrow~& (T, i) \models \phi \text{ or } (T, i) \models  \psi \\
			\overset{\text{IH}}{\Leftrightarrow}~& \emptyset, 0 \models_{T[r \mapsto s]} \exists r_1^1 \ldots r_n^1 \ldot \forall \pi \ldot [\phi, r] \text{ or } \emptyset, 0 \models_{T[r \mapsto s]} \exists r_1^2 \ldots r_m^2 \ldot \forall \pi \ldot [\psi, r]\\
			\Leftrightarrow~& \emptyset, 0 \models_{T[r \mapsto s]} \exists d^{\phi\clor\psi} \ldot \exists r_1^1 \ldots r_n^1 r_1^2 \ldots r_m^2 \ldot \forall \pi \ldot (d^{\phi\clor\psi}_\pi \rightarrow [\phi, r]) \land (\neg d^{\phi\clor\psi}_\pi \rightarrow [\psi,r]) 
		\end{align*}
		Case for $\phi \U \psi$:
		\begin{align*}
			&(T,i) \models \phi \U \psi\\
			&\Leftrightarrow \exists i' \geq i \ldot (T,i') \models \psi \text{ and } \forall i \leq i'' < i'\ldot (T, i'') \models \phi \\
			&\Leftrightarrow \exists i' \geq i\ldot (T,i') \models \psi \text{ and } \forall i \leq i'' < i'\ldot \emptyset, i'' \models_T \forall \pi \ldot \hat\phi \tag{by \Cref{thm:kcoherent-team-hyper}, since $\phi$ is 1-coherent} \\
			&\overset{\text{IH}}{\Leftrightarrow} \exists i' \geq i \ldot \emptyset,0 \models_{T[r^{\psi} \mapsto s^{\psi}]} \exists r_1 \ldots r_n \ldot \forall \pi \ldot [\psi, r^{\psi}]  \text{ and } \forall i \leq i'' < i'\ldot \emptyset, i'' \models_T \forall \pi \ldot \hat\phi \\
			& \tag{where $r^{\psi}$ is set in $s^{\psi}$ exactly at position $i'$} \\
			&\Leftrightarrow \exists i' \geq i \ldot \emptyset, 0 \models_{T[r \mapsto s, r^{\psi} \mapsto s^{\psi}, r^{\phi} \mapsto s^{\phi}]} \exists r_1 \ldots r_n \ldot \forall \pi \ldot \G(r^{\phi}_\pi \rightarrow \hat{\phi}) \land [\psi, r^{\psi}] \\
			&\tag*{(where $r^{\psi}$ is set in $s^{\psi}$ exactly at position $i'$ and} \\
			&\tag*{$r^{\phi}$ is set in $s^{\phi}$ exactly at all positions between $i$ and $i'$)} \\
			&\Leftrightarrow \emptyset, 0 \models_{T[r \mapsto s]} \models \exists r^{\psi} \ldot \exists r^{\phi} \ldot \exists r_1 \ldots r_n \ldot \forall \pi \ldot \G (r_\pi \rightarrow r^{\phi}_\pi \U (r^{\psi}_\pi \land \X \G \neg r^{\psi}_\pi)) \land \\
			&\phantom{\Leftrightarrow \emptyset, 0 \models_{T[r \mapsto s]}} \qquad \G(r^{\phi}_\pi \rightarrow \hat\phi) \land [\psi, r^{\psi}]
		\end{align*}
		Case for $\phi \W \psi$. Similar to $\phi \U \psi$.
\end{proof}


}{}

\end{document}